\documentclass[journal,twoside]{IEEEtran}

\pdfoutput=1

\usepackage{graphicx}              
\usepackage{amsmath}               
\usepackage{amsfonts}              
\usepackage{amsthm}               
\usepackage{amssymb}
\usepackage{cite}				  
\usepackage[acronym]{glossaries}    
\usepackage{dsfont}
\usepackage{xcolor}
\usepackage{microtype}
\usepackage{subcaption}

\newtheorem{thm}{Theorem}
\newtheorem*{thm*}{Theorem}
\newtheorem{lem}{Lemma}

\newtheorem{defn}{Definition}

\newcommand{\rr}{\mathrm}
\newcommand{\cc}{\mathcal}
\newcommand{\bb}{\mathbb}

\newacronym{avc}{AVC}{Arbitrarily Varying Channel}
\newacronym{avs}{AVS}{Arbitrarily Varying Source}
\newacronym{cr}{CR}{Common Randomness}
\newacronym{dmc}{DMC}{Discrete Memoryless Channel}
\newacronym{dms}{DMS}{Discrete Memoryless Source}
\newacronym{dmms}{DMMS}{Discrete Memoryless Multiple Source}
\newacronym{iid}{iid}{independently and identically distributed}
\newacronym{pd}{PD}{Probability Distribution}
\newacronym{rv}{RV}{Random Variable}
\newacronym{sk}{SK}{Secret-Key}
\newacronym{uwb}{UWB}{Ultra-Wideband}
\makeglossaries

\begin{document}

\title{Secret-Key Generation Using Compound Sources and One-Way Public Communication}

\author{Nima Tavangaran,~\IEEEmembership{Student Member,~IEEE,}
Holger Boche,~\IEEEmembership{Fellow,~IEEE,}
and Rafael F. Schaefer,~\IEEEmembership{Member,~IEEE}%
\thanks{This work was presented in part at IEEE-ITW, Jeju Island, Korea, 2015.}%
\thanks{N. Tavangaran and H. Boche are with Lehrstuhl f\"ur \mbox{Theoretische} Informationstechnik, Technische Universit\"at M\"unchen, 80290 M\"unchen, \mbox{Germany}, e-mail: nima.tavangaran@tum.de, boche@tum.de}%
\thanks{R. F. Schaefer was with the Department of Electrical Engineering, Princeton University, NJ 08544, USA. He is now with Technische Universit\"at Berlin, 10587 Berlin, Germany, email: rafael.schaefer@tu-berlin.de}%
\thanks{This work of N. Tavangaran and H. Boche was supported by the German Research Foundation (DFG) under grant BO 1734/20-1.}%
\thanks{This work of R. F. Schaefer was supported by the German Research Foundation (DFG) under grant WY 151/2-1.}%
}

\maketitle

\begin{abstract}
In the classical Secret-Key generation model, Common Randomness is generated by two terminals based on the observation of correlated components of a common source, while keeping it secret from a non-legitimate observer. It is assumed that the statistics of the source are known to all participants. In this work, the Secret-Key generation based on a compound source is studied where the realization of the source statistic is unknown. The protocol should guarantee the security and reliability of the generated Secret-Key, simultaneously for all possible realizations of the compound source. A single-letter lower-bound of the Secret-Key capacity for a finite compound source is derived as a function of the public communication rate constraint. A multi-letter capacity formula is further computed for a finite compound source for the case in which the public communication is unconstrained. Finally a single-letter capacity formula is derived for a degraded compound source with an arbitrary set of source states and a finite set of marginal states.
\end{abstract}


\section{Introduction}

Current cryptographic approaches are dependent on the computational capabilities of the non-legitimate terminals. By increasing technological advances, the security of transmitted information can not be guaranteed for sure. In contrast, an information theoretic approach provides us with a framework for future coding schemes that guarantee security independent of computational capabilities of the eavesdroppers. 

Information theoretic security was first introduced by Shannon in~\cite{sha_cts}. In the so called one-time pad method, each transmitting message is encrypted by a \gls{sk}.
If there is no \gls{sk} available, it has to be generated first. One approach is to generate a shared \gls{sk} based on a common source. In this model, two terminals observe correlated components of a common source and communicate over a public noiseless channel to generate a common \gls{sk}. Afterwards, they can encrypt subsequent communication using this \gls{sk}. This procedure relies on the generation of \gls{cr} which was introduced in~\cite{gac_cr0} and later used by Maurer in~\cite{mau_crm} and Ahlswede and Csisz\'ar in~\cite{ahl_crm} to determine the \gls{sk} capacity. 
The \gls{sk} sharing is further studied in~\cite{csi_crh,csi_smt,csk_inf,ign_bsi,lai_ps1}. 
In practice, this kind of security can be integrated in the physical layer of wireless systems~\cite{wil_suw}.

However, in all these models which were used for \gls{sk} generation, perfect knowledge of the source statistics was assumed. In a more general approach, the source uncertainty should be taken into account where the terminals do not have the knowledge of the actual realization of the source. An achievable \gls{sk} rate for a compound \gls{dmms} $\{(X,Y_s,Z_s)\}_{s\in\cc{S}}$ was given in~\cite{blo_chi}.
In~\cite{boc_skg}, the compound \gls{dmms} $\{(X_s,Y_s)\}_{s\in\cc{S}}$ was studied and the \gls{sk} capacity was computed.
As related problems, compound source coding, \gls{sk} generation with \gls{avc}, compound wiretap channels and robust biometric authentication were studied in~\cite{jan_dis,drp_csc,blo_skg,bje_scw,sch_gau,and_mas}.

In this work, a \gls{sk} generation model for a compound \gls{dmms} with one-way communication in presence of an eavesdropper is studied. The terminals observe a compound source 
$
\mathfrak{S}:=\{XYZ,s\}_{s\in\cc{S}}
:=\{(X_s,Y_s,Z_s)\}_{s\in\cc{S}}, 
$
and two of them generate a shared \gls{sk} by using only a one-way communication over a public noiseless channel while keeping it secret from the third terminal (eavesdropper). As the source realization index $s\in\cc{S}$ is unknown to the terminals, an estimation method such as hypothesis testing is incorporated to find the marginal source index of the transmitter. This approach is used to generalize the result in~\cite[Theorem 2.6]{csi_crh},\cite[Theorem 17.21]{csk_inf} to the compound setup for the source $\mathfrak{S}$.

In Section~\ref{sec_mod}, the model for \gls{sk} generation is presented. Section~\ref{sec_cap} gives the main results. A single-letter lower-bound for the \gls{sk} capacity of a finite compound source is derived as a function of the communication rate constraint over the public channel. A multi-letter \gls{sk} capacity formula is computed as well for the case where the public communication rate is unconstrained. We proposed these two results which are stated in Theorems 1 and 2, originally in~\cite{tav_skg} in which the proof ideas have been outlined. In the present paper, in contrast, we give the complete proofs of both theorems as well as the proofs for random coding and security lemmas.
The third and main result of this work which is not available in~\cite{tav_skg} gives a single-letter \gls{sk} capacity formula 
for a degraded compound source, where the set of source states may be infinite and the set of marginal states is finite. Compared with previous theorems, this result which is stated in Theorem~\ref{thm_inf} is more practical in the sense that the \gls{sk} capacity is single-letter and also valid for sources with an infinite set of source states $\cc{S}$. The complete formal proofs are given in Section~\ref{sec_prf}. Section~\ref{sec_con} concludes the paper.

\emph{Notation:}
$\bb{R}$ and $\bb{N}$ denote the set of real numbers and natural numbers respectively. The complement of a set $\cc{A}$ is denoted by $\cc{A}^{\rr{c}}$ and the subtraction of two sets $\cc{A}$ and $\cc{B}$ is given by $\cc{A}-\cc{B}:=\cc{A}\cap\cc{B}^{\,\rr{c}}$. The underlying probability measure is represented by $\rr{Pr}$. The $\log(\cdot)$ and $\exp(\cdot)$ are to the basis 2 and $\ln(\cdot)$ is the natural logarithm to the basis $\rr{e}$.  For any function $f$, the cardinal number of the range of the function is denoted by $\|f\|$. \glspl{rv} are denoted by capital letters \mbox{(e.g. $X^n,U,\cdots$)}, their realizations by small letters (e.g. $x^n,u,\cdots$), their ranges (alphabets) by script letters (e.g. $\cc{X}^n,\cc{U},\cdots$), and their \glspl{pd} by Roman letters (e.g. $\rr{P}_{X^n},\rr{P}_{U},\cdots$). All alphabets corresponding to \gls{rv}s are supposed to be finite. $H(X)$ and $I(X;Y)$ represent the entropy of a \gls{rv} $X$ and the mutual information between $X$ and $Y$ respectively. $h(a)$ with $a\in[0,1]$ is the binary entropy function and is given by \mbox{$h(a):=-a\log a - (1-a)\log(1-a)$}. For any two probability measures $\rr{P}$ and $\rr{Q}$, $\Vert \rr{P}-\rr{Q}\Vert:=\sum_{x\in\cc{X}}|\rr{P}(x)-\rr{Q}(x)|$ denotes their 1-norm distance. $\mathds{1}_\cc{A}(\cdot)$ denotes the indicator function for a set $\cc{A}$. $\bb{E}_U[X]$ represents the expectation of a \gls{rv} $X$ with respect to the \gls{rv} $U$. A stochastic matrix $W:\cc{X}\to\cc{P(Y)}$ is a function from the set $\cc{X}$ to $\cc{P(Y)}$, where $\cc{P(Y)}$ is the set of all \glspl{pd} defined on the set $\cc{Y}$. Finally, $X-Y-Z$ denotes a Markov chain for \gls{rv}s $X$, $Y$, and $Z$.

\section{\gls{sk} Generation Model}\label{sec_mod}

\begin{figure}[!tl]
\centering
\huge
\scalebox{0.5}{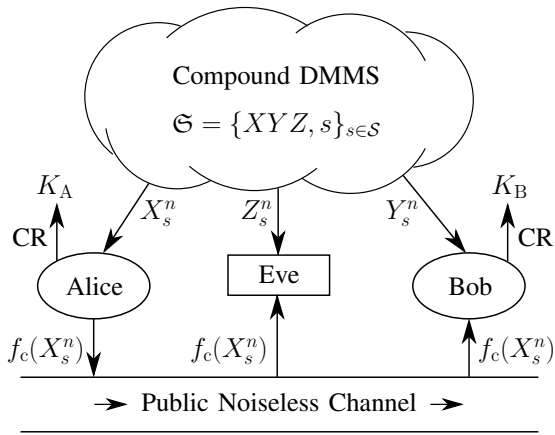}
\caption{\gls{sk} generation protocol for compound DMMS model}
\label{fig_sor}
\end{figure}

Figure~\ref{fig_sor} shows the \gls{sk} generation model which is used throughout this work. Transmitter (Alice), receiver (Bob) and eavesdropper (Eve) observe a compound \gls{dmms} \mbox{$\mathfrak{S}=\{XYZ,s\}_{s\in\cc{S}}$} for time duration $n\in \bb{N}$. It is assumed that all terminals know the set of source states $\cc{S}$ as well as its statistics with \glspl{pd} $\{\rr{P}_{XYZ,s}\}_{s\in\cc{S}}$. However, they do not have the knowledge of the actual realization $s\in\cc{S}$ of the source statistic. Therefore, \glspl{rv} $X_s^n,Y_s^n$ and $Z_s^n$ represent their initial knowledge for the source state $s\in \cc{S}$. The next definition describes the model which is studied through out this work.

\begin{defn}\label{def_csm}
The \gls{sk} generation model consists of a transmitter (Alice), a receiver (Bob), an eavesdropper (Eve), a compound \gls{dmms} by which their initial knowledge is given, and a public noiseless communication channel between all terminals. The source is given for an arbitrary set of states $\cc{S}$, by a sequence of generic \glspl{rv} $\mathfrak{S}=\{XYZ,s\}_{s\in\cc{S}}$ taking their values in the finite set $\cc{X}\times\cc{Y}\times\cc{Z}$. 
\end{defn}

As \glspl{rv} $X_s^n$ and $Y_s^n$ are correlated, Alice and Bob may generate some \gls{cr} by communicating over the public channel. In this work, only a one-way communication over the public channel, initiated by Alice, is allowed. The following definition gives a more precise description of this procedure.

\begin{defn}\label{def_pro}
A one-way \gls{sk} generation protocol for the model in Definition~\ref{def_csm} with source $\mathfrak{S}=\{XYZ,s\}_{s\in\cc{S}}$ consists of the following two steps:
\begin{itemize}
\item After observing $X_s^n$, Alice transmits a message $f_{\rr{c}}(X_s^n)$ to Bob over the public noiseless channel. $f_{\rr{c}}$ is called public communication function.\footnote{Similarly as in~\cite[Problem 17.15(a)]{csk_inf}, it can be shown that a randomized $f_{\rr{c}}$ in the one-way \gls{sk} generation protocol does not increase the \gls{sk} capacity. Therefore, the communication function $f_{\rr{c}}$ is assumed to be a deterministic function of $X_s^n$ and no randomization is considered here.}
\item Next, Alice generates a \gls{sk}, represented by a \gls{rv} $K_{\rr{A}}$, based on her knowledge $X_s^n$ and Bob generates a \gls{sk}, represented by a \gls{rv} $K_{\rr{B}}$, based on his knowledge $(Y_s^n, f_{\rr{c}}(X_s^n))$. $K_{\rr{A}}$ and $K_{\rr{B}}$ take their values in $\cc{K}$.
\end{itemize}
\end{defn}

As the communication over the public channel is overheard by Eve, this should not reveal any information about the \gls{sk}. 
Moreover, the generated \gls{sk} should have a uniform distribution. Combining these two criteria together leads to a compact notion, called security index, which was first introduced in~\cite{csi_smt}.
\begin{defn}\label{def_sec}
For \gls{rv}s $K_{\rr{A}}$ and $V$, taking values in the sets $\cc{K}$ and $\cc{V}$ respectively, the security index is given by 
\begin{align*}
S(K_{\rr{A}}|V):=\log(\cc{|K|})-H(K_{\rr{A}})+I(K_{\rr{A}};V).
\end{align*}
\end{defn}
In our context, $K_{\rr{A}}$ represents the \gls{sk} and $V$ Eve's knowledge. This short notion is a powerful tool which can be used to describe both strong 
secrecy~\cite{mau_str} and the uniformity of the generated \gls{sk}. The next definition, uses this concept to define an achievable \gls{sk} rate and capacity of a compound source. 
\begin{defn}\label{def_ach}
A real number $R_{\rr{sk}}\geq 0$ is an achievable \gls{sk} rate for the model in Definition~\ref{def_csm} with compound source \mbox{$\mathfrak{S}=\{XYZ,s\}_{s\in\cc{S}}$} and a one-way communication over the public noiseless channel with rate constraint $\Gamma\in(0,+\infty]$, if and only if, for all $\delta>0,$ and all $n\in\bb{N}$ large enough, there exists a \gls{sk} generation protocol with public communication function $f_{\rr{c}}$, giving rise to the \gls{rv}s $K_{\rr{A}}$ and $K_{\rr{B}}$ with values in $\cc{K}$, for which it holds:
\begin{align}
\frac{1}{n}\log\|f_{\rr{c}}\|<\Gamma+\delta,\label{eqn_ay1}\\
R_{\rr{sk}}<\frac{1}{n}\log|\cc{K}|+\delta,\label{eqn_ay2}\\
\forall s\in \cc{S},\quad \rr{Pr}(K_{\rr{A}}\neq K_{\rr{B}})<\delta,\label{eqn_ay3}\\
\forall s\in \cc{S},\quad S(K_{\rr{A}}|Z_s^n,f_{\rr{c}}(X_s^n))<\delta.\label{eqn_ay4}
\end{align}

The \gls{sk} capacity $C_{\rr{sk}}(\mathfrak{S},\Gamma)$ for this model is defined to be the supremum of all achievable \gls{sk} rates. If there is no communication rate constraint, i.e. $\Gamma=\infty,$ then condition~\eqref{eqn_ay1} in the definition is inactive and the capacity is simply denoted by $C_{\rr{sk}}(\mathfrak{S})$.
\end{defn} 

Similarly as in~\cite{csi_crh,csk_inf}, the communication rate constraint is also part of the achievability definition. This is because, in a realistic model where the communication cost is an important parameter, the information exchange rate between the terminals is restricted.

In the following, a subset of the compound set $\cc{S}$ is defined. This definition is required for stating the results in Section~\ref{sec_cap}.
\begin{defn}\label{def_set}
Let for the compound source $\mathfrak{S}=\{XYZ,s\}_{s \in \cc{S}}$, $\hat{\cc{S}}$ be the set of all possible states of marginal \gls{rv} $X$.
For a given marginal state $\hat{s}\in\cc{\hat{S}}$, corresponding to the \gls{rv} $X_{\hat{s}}$, the set of all possible source states is given by
\begin{align}\label{eqn_ixy}
\cc{I}(\hat{s}):=\Big\{s\in\cc{S}:&\forall x\in\cc{X},\nonumber\\
&\sum\limits_{y\in\cc{Y}} \sum\limits_{z\in\cc{Z}}\rr{P}_{XYZ,s}(x,y,z)=\rr{P}_{X_{\hat{s}}}(x)\Big\}.
\end{align}
\end{defn}

\section{\gls{sk} Capacity Results}\label{sec_cap}
In this section, \gls{sk} capacity results of compound \gls{dmms} models are presented. Moreover, short proof sketches are provided for the first and third theorems. The complete formal proofs are given in Section~\ref{sec_prf}. 
For all theorems which are stated in this section, the sets $\cc{\hat{S}}$ and $\cc{I}(\hat{s})$ are given by Definition~\ref{def_set}.

In the following, Theorem~\ref{thm_uv}
gives a single-letter lower-bound for the capacity of a compound \gls{dmms} source model with a finite set of source states where the public communication rate is limited.

\begin{thm}\label{thm_uv}
For a finite compound \gls{dmms} model with a source \mbox{$\mathfrak{S}=\{XYZ,s\}_{s\in\cc{S}}$} and a one-way communication over a public noiseless channel with constraint $\Gamma\in(0,\infty]$, it holds:
\begin{align}
C_{\rr{sk}}(\mathfrak{S},\Gamma)&\geq\min\limits_{\hat{s}
\in\hat{\cc{S}}}\max_{U_{\hat{s}},V_{\hat{s}}}\nonumber\\&\Big\{ \min\limits_{s\in\cc{I}(\hat{s})} I(V_{\hat{s}};Y_s|U_{\hat{s}})-\max\limits_{s\in\cc{I}(\hat{s})} I(V_{\hat{s}};Z_s|U_{\hat{s}}) \Big\},\label{eqn_thm} 
\end{align}
where the outer $\max$ is taken over all \gls{rv}s $U_{\hat{s}}$ and $V_{\hat{s}}$ such that it holds:
\begin{equation*}
\forall s\in\cc{I}(\hat{s}),\; U_{\hat{s}}-V_{\hat{s}}-X_{\hat{s}}-Y_sZ_s\quad\text{ and }
\end{equation*}
\begin{equation}\label{eqn_cns}
\max\limits_{s\in\cc{I}(\hat{s})} I(U_{\hat{s}};X_{\hat{s}}|Y_s) + \max\limits_{s\in\cc{I}(\hat{s})}I(V_{\hat{s}};X_{\hat{s}}|U_{\hat{s}}Y_s)< \Gamma.
\end{equation}
\end{thm}
\emph{Proof sketch:} To achieve the \gls{sk} rate in~\eqref{eqn_thm}, Alice estimates her marginal state $\hat{s}\in\cc{\hat{S}}$ by hypothesis testing such that the estimation error is exponentially small~\cite{hof_tes}. Similarly as in~\cite{boc_skg}, she sends this along with other information related to her observation over the public channel to Bob. In Figure~\ref{fig_sor}, this is denoted by $f_{\rr{c}}(X_s^n)$. Given an estimated marginal source state of Alice $\hat{s}\in\cc{\hat{S}}\,$, the joint source state $s$ is not necessarily known to the terminals. However, by Definition~\ref{def_set}, it is known that $s\in\cc{I}(\hat{s})$. 

For the correctly estimated Alice's state, say $\hat{s}\in\cc{\hat{S}}$, Lemma~\ref{lem_rac} from Subsection~\ref{ssc_pre}, assures that Alice 
and Bob can generate some \gls{cr} by using their knowledge with an exponentially small error. This \gls{cr} is universal for all source states $s\in\cc{I}(\hat{s})$.

Furthermore, the public communication rate should be lower than a given $\Gamma>0$. Therefore, the coding scheme, introduced in Lemma~\ref{lem_rac}, should work with respect to this limitation which is given by~\eqref{eqn_cns}. This problem for the case where no communication constraint is given is easier to solve and a non-compound version is available in~\cite[Problem 17.15b]{csk_inf}.

Finally, as seen in Figure~\ref{fig_sor}, Alice and Bob generate their \glspl{sk} $K_{\rr{A}}$ and $K_{\rr{B}}$ based on the \gls{cr} by using a \gls{sk} generator. However, $f_{\rr{c}}(X_s^n)$ is also received by Eve. 
Lemma~\ref{lem_skg}, again from Subsection~\ref{ssc_pre}, assures the existence of a \gls{sk} generator which guarantees the strong secrecy and the uniformity of the \gls{sk} $K_{\rr{A}}$, for all possible $s\in\cc{I}(\hat{s})$. \qed
 
As a second result, a multi-letter \gls{sk} capacity formula is computed for the case, in which no communication rate constraint is given and the set of source states is again finite.

\begin{thm}\label{thm_uci}
For a finite compound \gls{dmms} model with a source \mbox{$\mathfrak{S}=\{XYZ,s\}_{s\in\cc{S}}$} and a one-way communication over a public noiseless channel, it holds:
\begin{align}
C_{\rr{sk}}(\mathfrak{S})&=\lim_{n\rightarrow\infty}\frac{1}{n}\min\limits_{\hat{s}\in\hat{\cc{S}}}\max_{U_{\hat{s}},V_{\hat{s}}}\nonumber\\&\Big\lbrace \min\limits_{s\in\cc{I}(\hat{s})} I(V_{\hat{s}};Y_s^n|U_{\hat{s}})
-\max\limits_{s\in\cc{I}(\hat{s})} I(V_{\hat{s}};Z_s^n|U_{\hat{s}}) \Big\rbrace,\label{eqn_tci} 
\end{align}
where the outer $\max$ is taken over all \gls{rv}s $U_{\hat{s}}$ and $V_{\hat{s}}$ such that it holds:
\begin{equation}\label{eqn_ma2}
\forall s\in\cc{I}(\hat{s}),\; U_{\hat{s}}-V_{\hat{s}}-X_{\hat{s}}^n-Y_s^nZ_s^n.
\end{equation}
\end{thm}

Finally as the main result, a single-letter \gls{sk} capacity formula is given in the following for a degraded compound source with an arbitrary set of source states $\cc{S}$ which might be infinite.

\begin{thm}\label{thm_inf}
Consider a compound \gls{dmms} model with a source \mbox{$\mathfrak{S}=\{XYZ,s\}_{s\in\cc{S}}$} with an arbitrary set $\cc{S}$, a finite set of marginal states $\cc{\hat{S}}$, and a one-way communication over a public noiseless channel. 
If the following Markov chains are satisfied,
\begin{align}\label{eqn_if1}
\forall\hat{s}\in\cc{\hat{S}},\;\forall r,t\in\cc{I}(\hat{s}),\quad X_{\hat{s}}-Y_r-Z_t,
\end{align}
then it holds:
\begin{align}
C_{\rr{sk}}(\mathfrak{S})&=\min_{\hat{s}\in\cc{\hat{S}}}\Big\lbrace\inf_{r\in\cc{I}(\hat{s})}I(X_{\hat{s}};Y_r)-\sup_{t\in\cc{I}(\hat{s})}I(X_{\hat{s}};Z_t)\Big\rbrace.\label{eqn_if2}
\end{align}
\end{thm}
\emph{Proof sketch:}
In the first step, the achievability result from Theorem~\ref{thm_uv} is used to show that the \gls{sk} rate in~\eqref{eqn_if2} is achievable for a finite set of source states $\cc{S}$. 

Next, for the infinite source $\mathfrak{S}=\{XYZ,s\}_{s\in\cc{S}}$, fix the marginal \gls{pd} $\rr{P}_{X_{\hat{s}}}$ and define the infinite family of stochastic matrices $\{\rr{P}_{YZ,s|X_{\hat{s}}}:\cc{X}\to\cc{P(Y}\times \cc{Z})\}_{s\in\cc{I}(\hat{s})}$. By using Lemma~\ref{lem_app} from Subsection~\ref{ssc_pre}, it follows that there exists a finite family of stochastic matrices $\{\rr{W}_{s'}:\cc{X}\to\cc{P(Y}\times \cc{Z)}\}_{s'\in\cc{I}'(\hat{s})}$, which approximates the infinite family. The finite set $\cc{I}'(\hat{s})$ contains all indices of this finite family for the given $\hat{s}$.

Define the finite source $\mathfrak{S}':=\{XYZ,s'\}_{s'\in\cc{I}'(\hat{s}),\hat{s}\in\hat{\cc{S}}}$ by $\rr{P}_{XYZ,s'}(x,y,z):=\rr{P}_{X_{\hat{s}}}(x)\rr{W}_{s'}(y,z|x)$ for all $s'\in\cc{I}'(\hat{s}),$ $\hat{s}\in\hat{\cc{S}},$ and $(x,y,z)\in\cc{X\times Y\times Z}$. 
It is shown that the \gls{sk} generation protocol, which is used for the finite source $\mathfrak{S}'$, also guarantees the achievability of the given rate in~\eqref{eqn_if2} for the infinite source $\mathfrak{S}$.
\qed

The orders of Markov chains in Theorem~\ref{thm_inf} are crucial in determining the \gls{sk} capacity. For example, assume that the orders of~\eqref{eqn_if1} are changed as given in the following,
\begin{align*}
\forall\hat{s}\in\cc{\hat{S}},\;\forall r,t\in\cc{I}(\hat{s}),\quad X_{\hat{s}}-Z_t-Y_r.
\end{align*}
In this case the capacity is $C_{\rr{sk}}(\mathfrak{S})=0$. This is because, by~\cite[Theorem 1]{ahl_crm}, it holds that
\begin{align*}
\forall\hat{s}\in\hat{S},\;\forall r,t\in\cc{I}(\hat{s}), \; C_{\rr{sk}}(\mathfrak{S})\leq I(X_{\hat{s}};Y_r|Z_t)=0.
\end{align*}

\section{Proofs}\label{sec_prf}
This section is divided into 3 parts. Subsection~\ref{ssc_pre} gives a short review of definitions and results which are required in the proofs. In Subsection~\ref{ssc_rns}, Lemmas~\ref{lem_rac} and~\ref{lem_skg} for random coding and security are presented. Finally, Subsection~\ref{ssc_mai} presents the formal proofs of Theorems~\ref{thm_uv},~\ref{thm_uci}, and~\ref{thm_inf}.

\subsection{Preliminaries}\label{ssc_pre}

For the typical sequences and their related sets the same definitions as in~\cite[Chapters 2 and 17]{csk_inf} are taken.
Let $N(x|x^n)$ give the number of repetitions of an element $x$ in the sequence $x^n$ and $N(x,y|x^n,y^n)$  the number of repetitions of the pair $(x,y)$ in the pair sequence $(x^n,y^n)$. For two \gls{rv}s $X$ and $Y$ and stochastic matrix \mbox{$\rr{P}_{Y|X}:
\cc{X}\rightarrow\cc{P(Y)}$}, the following definitions are used when $\epsilon>0$:
\begin{align*}
&\cc{T}_{[XY]\epsilon}^n := \Big\{(x^n,y^n)\in\cc{X}^n\times\cc{Y}^n : \forall (x,y)\in\cc{X}\times\cc{Y},\\ 
&\qquad\qquad\qquad\big|\rr{P}_{XY}(x,y)-\frac{1}{n}N(x,y|x^n,y^n)\big|\leq \epsilon\\ 
&\qquad\qquad\qquad\wedge \big( \rr{P}_{XY}(x,y)=0 \Rightarrow N(x,y|x^n,y^n)=0 \big) \Big\},\\
&\cc{T}_{[Y|X]\epsilon}^n(x^n) := \Big\{y^n\in\cc{Y}^n : \forall (x,y)\in\cc{X}\times\cc{Y},\\ 
&\qquad\qquad\quad\big|\frac{1}{n}N(x|x^n)\rr{P}_{Y|X}(y|x)-\frac{1}{n}N(x,y|x^n,y^n)\big|\leq \epsilon\\ 
&\qquad\qquad\quad\wedge \big( \rr{P}_{Y|X}(y|x)=0 \Rightarrow N(x,y|x^n,y^n)=0 \big) \Big\},\\
&\cc{T}_{[XY]\epsilon}^n(x^n) := \Big\{y^n\in\cc{X}^n : (x^n,y^n)\in\cc{T}_{[XY]\epsilon}^n \Big\}.
\end{align*}

In the following, a series of lemmas and propositions is provided which will be used in the proofs. 
\begin{lem}[\!\!\cite{csk_inf}]\label{lem_typ}
Let $U,X,$ and $Y$ be \gls{rv}s taking value in $\cc{U}, \cc{X},$ and $\cc{Y}$ respectively. Assume $0<\xi<\zeta<\sigma,$ and $\tau>0$ are all in $\bb{R}$ and $n\in\bb{N}$. Then, it holds that
\begin{enumerate}
\itemsep0.8em
\item\label{typ_sub} $\forall x^n\in \cc{T}_{[X]\xi}^n$, $\cc{T}_{[XY]\zeta}^n (x^n)\supset \cc{T}_{[Y|X]\zeta-\xi}^n(x^n).$

\item\label{typ_ptx} $\rr{P}_X^n(\cc{T}_{[X]\xi}^n)\geq 1-2|\cc{X}|\rr{e}^{-2\xi^2 n}.$

\item\label{typ_px1} $\forall\tau>0$, $\forall\xi>0$ sufficiently small, and $\forall x^n\in\cc{T}^n_{[X]\xi}$, it holds $\big| -\frac{1}{n}\log\rr{P}^n_X(x^n) -H(X)\big|<\tau.$

\item \label{typ_ca1} $\forall\tau>0$, $\forall\xi>0$ sufficiently small, and $\forall n\in\bb{N}$ sufficiently large, it holds $\big| \frac{1}{n}\log |\cc{T}_{[X]\xi}^n| -H(X)\big|<\tau.$

\item \label{typ_ca2} $\forall\tau>0$, $\forall\zeta>0$ sufficiently small, $\forall n\in\bb{N}$ sufficiently large, and $\forall u^n$ with $\cc{T}_{[UX]\zeta}^n(u^n)\neq\emptyset$ it holds $\big| \frac{1}{n}\log |\cc{T}_{[UX]\zeta}^n(u^n)| -H(X|U)\big|<\tau.$

\item\label{typ_mar} For $U-X-Y$ and $\forall(u^n,x^n)\in\cc{T}_{[UX]\xi}^n$ it holds\\
$\rr{P}_{\!Y|X}^n\big(\cc{T}_{[UXY]\sigma}^n(u^n\!\!,x^n)|x^n\big)\!\geq\! 1-2\cc{|U||X||Y|}\rr{e}^{-2(\sigma-\xi)^2 n}.$

\item\label{typ_ixy} $\forall\tau>0$, $\forall\zeta>0$ sufficiently small, $\forall n\in\bb{N}$ sufficiently large, and $\forall y^n\in\cc{Y}^n$ if $\cc{T}^n_{[XY]\zeta}(y^n)\neq\emptyset$ then it holds
 $\Big|-\frac{1}{n}\log \rr{P}_X^n\big(\cc{T}^n_{[XY]\zeta}(y^n)\big)-I(X;Y)\Big|<\tau.$
 
\item\label{typ_las} $\forall\tau>0$, $\forall\zeta>0$ sufficiently small, $\forall n\in\bb{N}$ sufficiently large, and $\forall (u^n,y^n)\in\cc{U}^n\times\cc{Y}^n$ if $\cc{T}^n_{[UXY]\zeta}(u^n,y^n)\neq\emptyset$ then it holds\\
$\big|-\frac{1}{n}\log \rr{P}_{X|U}^n\big(\cc{T}^n_{[UXY]\zeta}(u^n,y^n)|u^n\big)-I(X;Y|U)\big|<\tau.$
\end{enumerate}
\end{lem}

\begin{lem}[\!\!\cite{csk_inf}]\label{lem_ran}
Let \gls{rv}s $U$ and $X$ take their values in $\cc{U}$ and $\cc{X}$ respectively. Consider $N=\exp(n R)$ sequences $u_l^n$ with $l\in\cc{L}:=\{1,2\cdots N\}$, which are independently drawn by a \gls{pd} $\rr{P}_U^n$ with $I(U;X)<R$. Then for all $\tau\in (0,R-I(U;X))$, all $\zeta>0$ sufficiently small, all $n\in \bb{N}$ sufficiently large, and all $x^n$ such that $\cc{T}_{[UX]\zeta}^n(x^n)\neq \emptyset$, it holds that
\begin{align*}
&\bigg\vert\frac{1}{n}\log\Big|\big\{l\in\cc{L}:u_l^n\in\cc{T}_{[UX]\zeta}^n(x^n)\big\}\Big|-\big(R-I(U;X)\big)\bigg\vert<\tau,
\end{align*}
with a probability approaching one, doubly exponentially fast.
\end{lem}

\begin{lem}[\!\cite{csk_inf}]\label{lem_ra2}
Let \gls{rv}s $U,V,$ and $X$ take their values in $\cc{U},\cc{V},$ and $\cc{X}$ respectively. Consider $N=\exp(n R)$ sequences $v_l^n$ with $l\in\cc{L}:=\{1,2\cdots N\}$, which are independently drawn by a \gls{pd} $\rr{P}_{V|U}^n(.|u^n)$ for a given $u^n$ with $I(V;X|U)<R$. Then for all $\tau\in (0,R-I(V;X|U))$, all $\sigma>0$ sufficiently small, all $n\in \bb{N}$ sufficiently large, and all $x^n$ such that $\cc{T}_{[UVX]\sigma}^n(u^n,x^n)\neq \emptyset$, it holds that
\begin{align*}
&\bigg\vert\frac{1}{n}\!\log\Big|\big\{l\!\in\!\cc{L}\!:\! v_l^n\!\!\in\!\cc{T}_{[UVX]\sigma}^n(u^n\!\!,x^n)\big\}\Big| \!-\!\big(R\!-\!I(V;\!X|U)\big)\bigg\vert\!<\!\tau,
\end{align*}
with a probability approaching one, doubly exponentially fast. This holds uniformly for all given $u^n\in\cc{U}^n.$
\end{lem}

\begin{lem}[\!\!\cite{ahl_cr2,csk_inf}]\label{lem_ext}
Let $\epsilon,\eta>0$ and $\lambda>0$ all in $\bb{R}$ and $k\in\bb{N}$ be given and U be a \gls{rv} taking value in $\cc{U}$. If $\rr{P}_U(\{u\in\cc{U}:\rr{P}_U(u)\leq\frac{1}{\lambda}\})\geq 1-\eta$, then for a randomly selected function $\kappa:\cc{U}\rightarrow\{1,2,\cdots,k\}$ it holds, \[\rr{Pr}\Big(\big\|\kappa(\rr{P}_U)-\rr{P}_0\big\|>\epsilon+2\eta\Big)\leq 2k\rr{e}^{-\frac{\lambda\epsilon^2(1-\eta)}{2k(1+\epsilon)}},\]
where $P_0(i)=1/k$ for all $i=1,2,\cdots k$.
Random selection means that the $\kappa(u),u\in\cc{U}$ are chosen \gls{iid} uniformly.
\end{lem}

\begin{lem}[\!\!\cite{blk_cap}]\label{lem_app}
Let $\cc{S}$ be an arbitrary set and possibly infinite and $\{\rr{W}_s:\cc{X}\to\cc{P(U)}\}_{s\in\cc{S}}$ be a family of stochastic matrices. For every $l\in\bb{N}$ with $l\geq 2|\cc{U}|^2$, there exists a family of stochastic matrices $\{\rr{W}_{s'}:\cc{X}\to\cc{P(U)}\}_{s'\in\cc{S}'}$ with a finite set $\cc{S}'$, such that $|\cc{S}'|\leq (l+1)^{|\cc{X\times U}|}$, where the following properties hold: $\forall s\in\cc{S},\;\exists s'\in\cc{S}',\forall x\in\cc{X},\forall u\in\cc{U},$
\begin{align*}
 |\rr{W}_s(u|x)\!-\!\rr{W}_{s'}(u|x)|\!\leq \frac{1}{l}|\cc{U}|,\;
\rr{W}_s(u|x)\leq \rr{e}^{2|\cc{U}|^2/l}\rr{W}_{s'}(u|x).
\end{align*}
\end{lem}

\begin{lem}[\!\!\cite{zha_est}]\label{lem_muc}
Let $(X,Y)$ and $(X',Y')$ be two pairs of \glspl{rv} taking values in $\cc{X}\times\cc{Y}$ with \glspl{pd} $\rr{P}_{XY}$ and $\rr{P}_{X'Y'}$ respectively. Furthermore, let $\gamma:=\frac{1}{2}\Vert \rr{P}_{XY}-\rr{P}_{X'Y'}\Vert$ and $\gamma\leq 1 - \frac{1}{|\cc{X}\times\cc{Y}|}.$
Then, it holds that
\begin{align*}
|I(X;Y)-I(X';Y')|\leq 3\gamma\log(|\cc{X}\times\cc{Y}|-1)+3h(\gamma).
\end{align*}
\end{lem}

\subsection{Random Coding and Security Lemmas}\label{ssc_rns}

In this subsection, Lemmas~\ref{lem_rac} and~\ref{lem_skg} for  finite compound sets with their proofs are presented. Similar techniques which are used in the non-compound versions in~\cite[Lemmas 17.5 and 17.22]{csk_inf},\cite{csi_alm,csi_crh}, are used in the following proofs and extended to the compound setup. For completeness, we present all proofs in detail. These lemmas are required for the proof of Theorem~\ref{thm_uv}.
Assume in Lemma~\ref{lem_rac}, if the values in the equations~\eqref{eqn_n1}, \eqref{eqn_n2}, \eqref{eqn_n3}, and~\eqref{eqn_n4} are not integer numbers, then the smallest integer which is larger than the given expression is taken.

\begin{lem}\label{lem_rac}
Let $\delta>0$ and $\sigma>\zeta>0$ be all in $\bb{R}$ and sufficiently small. Furthermore, let Alice's state $\hat{s}\in\hat{\cc{S}}$ be given and for all $s\in\cc{I}(\hat{s})$, \glspl{rv} $X_{\hat{s}}$ and $Y_s$ take their values in $\cc{X}$ and $\cc{Y}$ respectively. Let also \glspl{rv} $U_{\hat{s}}$ and $V_{\hat{s}}$ be given such that for all $s\in\cc{I}(\hat{s})$ the Markov chains \mbox{$U_{\hat{s}}-V_{\hat{s}}-X_{\hat{s}}-Y_s$} hold. 

Assume $N_{\hat{s},1} N_{\hat{s},2}$ random sequences $u_{i j}^n(\hat{s})\in\cc{U}^n$, chosen independently according to \gls{pd} $\rr{P}_{U_{\hat{s}}}^n$, are given and known to Alice and Bob where
\begin{align}
&i\in \cc{I}:=\big\{1,2,\cdots,N_{\hat{s},1}\big\},\;j\in \cc{J}:=\big\{1,2,\cdots,N_{\hat{s},2}\big\},\nonumber\\
&N_{\hat{s},1}:=\exp\Big[n\big(\max\limits_{s\in{\cc{I}(\hat{s})}}I(U_{\hat{s}};X_{\hat{s}}|Y_s)+3\delta\big)\Big],\label{eqn_n1}\\
&N_{\hat{s},2}:=\exp\Big[n\big(\min\limits_{s\in{\cc{I}(\hat{s})}}I(U_{\hat{s}};Y_s)-2\delta\big)\Big].\label{eqn_n2}
\end{align}

Moreover, let for each $u_{i j}^n(\hat{s})$, $N_{\hat{s},3}N_{\hat{s},4}$ random sequences ${v^{ij}_{pq}}^n(\hat{s})\in\cc{V}^n$, chosen conditionally independently according to $\rr{P}^n_{V_{\hat{s}}|U_{\hat{s}}}\!(\cdot|u_{i j}^n(\hat{s}))$, be given and known to Alice and Bob where
\begin{align}
&p\in \cc{P}:=\big\{1,2,\cdots,N_{\hat{s},3}\big\},\;q\in \cc{Q}:=\big\{1,2,\cdots,N_{\hat{s},4}\big\},\nonumber\\
&N_{\hat{s},3}:=\exp\Big[n\big(\max\limits_{s\in{\cc{I}(\hat{s})}}I(V_{\hat{s}};X_{\hat{s}}|U_{\hat{s}}Y_s)+3\delta\big)\Big],\label{eqn_n3}\\
&N_{\hat{s},4}:=\exp\Big[n\big(\min\limits_{s\in{\cc{I}(\hat{s})}}I(V_{\hat{s}};Y_s|U_{\hat{s}})-2\delta\big)\Big].\label{eqn_n4}
\end{align}
Then \gls{cr} can be generated between Alice and Bob in two ways:

\vspace*{0.5em}
a) For $n\in\bb{N}$ sufficiently large, there exist encoder functions \mbox{$f:\cc{T}\rightarrow\cc{I}$} and $g:\cc{T}\rightarrow\cc{J}$, with a probability approaching 1, doubly exponentially fast where 
\begin{align}\label{eqn_tfg}
\cc{T}:=\Big\{x^n\in\cc{X}^n:\cc{T}_{[U X,{\hat{s}}]\zeta}^n(x^n)\neq\emptyset\Big\},
\end{align}
and if $f(x^n)=i, \;g(x^n)=j$ then 
$
(u_{i j}^n(\hat{s}),x^n)\in\cc{T}_{[UX,{\hat{s}}]\zeta}^n.
$
Alice encodes her observation $x^n\in\cc{T}$ by these functions to the sequence $u_{i j}^n(\hat{s})$ where $j$ is the \gls{cr}.

For functions $f$ and $g$, extending their domain to $\cc{X}^n$ by defining for all $x^n\not\in\cc{T},\;f(x^n)=g(x^n)=0$, there exists a decoder $\tilde{g}:\cc{I}\times\cc{\hat{S}}\times\cc{Y}^n\rightarrow\cc{J}$ such that for all $s\in\cc{I}(\hat{s})$,
\begin{align}\label{eqn_es1}
\rr{Pr}\Big\lbrace g(X_{\hat{s}}^n) \neq \tilde{g}\big(f(X_{\hat{s}}^n),\hat{s},Y_s^n\big)\Big\rbrace<\exp(-n\delta_0),
\end{align}
for some $\delta_0>0$. Thus, Bob can reconstruct $g(x^n)=j$ from $(f(x^n),\hat{s},y^n)$ for given realizations $X_{\hat{s}}^n=x^n$ and $Y_s^n=y^n$.

\vspace*{0.5em}
b) For each $f$ and $g$ from part a), and $n\in\bb{N}$ sufficiently large, there exist encoder functions $\varphi:\cc{T}\rightarrow\cc{P}$ and \mbox{$\rho:\cc{T}\rightarrow\cc{Q}$} with a probability approaching 1, doubly exponentially fast, such that if $f(x^n)\!=\!i,\, g(x^n)\!=\!j,\, \varphi(x^n)\!=\!p,\,\rho(x^n)\!=\!q$ then 
$
(u_{i j}^n(\hat{s}),{v^{ij}_{pq}}^n(\hat{s}),x^n)\in\cc{T}_
{[UVX,{\hat{s}}]\sigma}^n.
$
Alice encodes her observation $x^n\in\cc{T}$ by these functions to the sequence ${v^{ij}_{pq}}^n(\hat{s})$ where $q$ is the \gls{cr}.

For functions $\varphi$ and $\rho$, extending their domain to $\cc{X}^n$ by defining
for all $x^n\not\in\cc{T},\; \varphi(x^n)=\rho(x^n)=0$,
there exists a decoder \mbox{$\tilde{\rho}:\cc{I}\times\cc{J}\times\cc{P}\times\cc{\hat{S}}\times\cc{Y}^n\rightarrow\cc{Q}$}
such that for all $s\in\cc{I}(\hat{s})$,
\begin{align}
\rr{Pr}\Big\lbrace \rho(X_{\hat{s}}^n) \neq \tilde{\rho}\big(f(X_{\hat{s}}^n),g(X_{\hat{s}}^n),\varphi(&X_{\hat{s}}^n),\hat{s},Y_s^n\big)\Big\rbrace\nonumber\\&<\exp(-n\delta'_0),\label{eqn_es2}
\end{align}
for some $\delta'_0>0$. Thus, Bob can reconstruct $\rho(x^n)=q$ from $(f(x^n),g(x^n),\varphi(x^n),\hat{s},y^n)$ for given realizations $X_{\hat{s}}^n=x^n$ and $Y_s^n=y^n$.
\end{lem}

\begin{proof}[Proof]
a) Let $R$ be the rate of choosing the sequences $\{u_{i j}^n(\hat{s})\}_{(i,j)\in\cc{I}\times\cc{J}}$ which implies that \mbox{$N_{\hat{s},1} N_{\hat{s},2} = \exp(nR)$}. Therefore, by~\eqref{eqn_n1} and~\eqref{eqn_n2} and properties of the Markov chain, it follows that
\begin{align}
R &= \max_{s\in\cc{I}(\hat{s})}I(U_{\hat{s}};X_{\hat{s}}|Y_s) + \min_{s\in\cc{I}(\hat{s})}I(U_{\hat{s}};Y_s)+\delta\nonumber\\
&= I(U_{\hat{s}};X_{\hat{s}}) + \max_{s\in\cc{I}(\hat{s})}[I(U_{\hat{s}};Y_s|X_{\hat{s}})-I(U_{\hat{s}};Y_s)] \nonumber\\ 
&\quad+\min_{s\in\cc{I}(\hat{s})}I(U_{\hat{s}};Y_s)+\delta = I(U_{\hat{s}};X_{\hat{s}})+\delta.\label{eqn_rat}
\end{align}

Similarly as in~\cite[Lemma 17.22]{csk_inf}, for all $x^n\in\cc{T},$ it holds by the definition in~\eqref{eqn_tfg} that $\cc{T}_{[UX,{\hat{s}}]\zeta}^n(x^n)\neq\emptyset.$ Thus, Lemma~\ref{lem_ran} together with~\eqref{eqn_rat} implies
for all $\tau\in (0,R-I(U_{\hat{s}};X_{\hat{s}}))$, all $\zeta>0$ sufficiently small, all $n\in \bb{N}$ sufficiently large, and all $x^n\in\cc{T}$ that
\begin{align*}
\frac{1}{n}\log\Big|\Big\{(i,j)\in\cc{I}\times\cc{J}: u_{i j}^n(\hat{s})\in&\cc{T}_{[UX,{\hat{s}}]\zeta}^n(x^n)\Big\}\Big|\nonumber\\ 
&\quad\geq R-I(U_{\hat{s}};X_{\hat{s}})-\tau,
\end{align*}
with a probability approaching one, doubly exponentially fast. Therefore, for each $x^n\in\cc{T}$, the number of chosen sequences $u_{i j}^n(\hat{s})$ which are in  $\cc{T}_{[UX,{\hat{s}}]\zeta}^n(x^n)$ is non-zero and the functions $f$ and $g$ as mentioned in the lemma, do exist with this probability.

Define for all $i\in\cc{I}$ and $y^n\in\cc{Y}^n$, the decoder as follows:
\begin{align}
\tilde{g}(i,\hat{s},y^n)\!:=\!
\begin{cases}
j & \text{if } j\!\in\!\cc{J}\!, u_{i j}^n(\hat{s}) \!\in\!\!\!\!\mathop{\bigcup}\limits_{s\in\cc{I}(\hat{s})}\!\!\!\!\cc{T}^n_{[UY,s]\sigma|\cc{X}|}(y^n)\\ &\;\text{and }\;\forall m\in\cc{J}, m\neq j\Rightarrow\\ &\quad\;\; u_{im}^n(\hat{s}) \not\in\!\!\!\mathop{\bigcup}\limits_{s\in\cc{I}(\hat{s})}\!\!\!\cc{T}^n_{[UY,s]\sigma|\cc{X}|}(y^n)\\
0 & \text{otherwise.}\label{eqn_dec}
\end{cases}
\end{align}
Moreover, define 
\begin{align*}
\cc{T}_0&:=\Big\lbrace(x^n,y^n)\in\cc{X}^n\!\times\!\cc{Y}^n: x^n\in\cc{T}\nonumber\\ &\qquad\wedge(u_{f(x^n)g(x^n)}^n(\hat{s}),x^n,y^n)\in\mathop{\bigcup}\limits_{s\in\cc{I}(\hat{s})}\cc{T}^n_{[UXY,s]\sigma}\Big\rbrace.
\end{align*}

In the following, it is shown that Alice and Bob's observation $(x^n,y^n)\in\cc{X}^n\times\cc{Y}^n$ is in the set $\cc{T}_0$ with a probability exponentially close to one. It holds that
\begin{align}
&\rr{P}_{XY,s}^n({\cc{T}_0}^{\rr{c}}) \nonumber\\ &=\sum\limits_{x^n\in\cc{T}^{\rr{c}},y^n\in\cc{Y}^n}\!\!\!\!\!\!\rr{P}_{XY,s}^n(x^n,y^n) + \!\!\!\!\sum\limits_{x^n\in\cc{T},(x^n,y^n)\not\in\cc{T}_0}\!\!\!\!\!\!\rr{P}_{XY,s}^n(x^n,y^n) \nonumber\\
&=\rr{P}_{X_{\hat{s}}}^n(\cc{T}^{\rr{c}})+
\sum\limits_{x^n\in\cc{T},(x^n,y^n)\not\in\cc{T}_0}\rr{P}_{X_{\hat{s}}}^n(x^n)\nonumber\\
&\quad\times\rr{P}_{Y_s| X_{\hat{s}}}^n\Big(\!\mathop{\bigcap}\limits_{s\in\cc{I}(\hat{s})}\!\!{\cc{T}^n_{[UXY,s]\sigma}}^{\!\!\!\!\!\!\!\rr{c}}\;\;(u_{f(x^n)g(x^n)}^n(\hat{s}),x^n)\big|x^n\Big).\label{eqn_err}
\end{align}
Lemma~\ref{lem_typ}.\ref{typ_sub} implies for every $x^n\in\cc{T}_{[X_{\hat{s}}]\xi}^n$ with $\xi\in(0,\zeta)$ that 
$\cc{T}_{[UX,\hat{s}]\zeta}^n (x^n)\neq\emptyset$ for $n$ large enough. Thus by the definition in~\eqref{eqn_tfg}, it follows that $\cc{T}\supset\cc{T}_{[X_{\hat{s}}]\xi}^n$ and thus Lemma~\ref{lem_typ}.\ref{typ_ptx} implies for some $c_0>0$ that
\begin{align}\label{eqn_l1}
\rr{P}_{X_{\hat{s}}}^n(\cc{T}^{\rr{c}})<\exp(-nc_0).
\end{align} 

On the other hand, for every $x^n\in\cc{T}$ with $f(x^n)=i$ and $g(x^n)=j$, it holds that $(u_{i j}^n(\hat{s}),x^n)\in\cc{T}_{[UX,\hat{s}]\zeta}^n$. Since for all $s\in\cc{I}(\hat{s}),$ the Markov chains $U_{\hat{s}}-X_{\hat{s}}-Y_s$ hold, Lemma~\ref{lem_typ}.\ref{typ_mar} implies for $c_0$ sufficiently small that 
\begin{align*}
\rr{P}_{Y_s|X_{\hat{s}}}^n\big({\cc{T}^n_{[UXY,s]\sigma}}^{\!\!\!\!\!\!\!\rr{c}}\;\;(u_{f(x^n)g(x^n)}^n(\hat{s}),x^n)|x^n\big)<\exp(-nc_0).
\end{align*}
This inequality together with~\eqref{eqn_err} and~\eqref{eqn_l1} gives
\begin{align}
\rr{P}_{XY,s}^n(\cc{T}_0^{\rr{c}})<\exp(-nc_1),\label{eqn_lfn}
\end{align} 
for some $c_1>0$ and $n$ sufficiently large.

Therefore, to compute the upper-bound of the probability in~\eqref{eqn_es1}, we may just concentrate on all $(x^n,y^n)\in\cc{T}_0$ with 
\begin{align}\label{eqn_nec}
\tilde{g}(f(x^n),\hat{s},y^n)\neq g(x^n). 
\end{align}

\noindent A necessary condition for $(x^n,y^n)\in\cc{T}_0$ is given by
\begin{align*}
(u_{f(x^n)g(x^n)}^n(\hat{s}),y^n)\in\mathop{\bigcup}\limits_{s\in\cc{I}(\hat{s})}\cc{T}_{[UY,s]\sigma|\cc{X}|}^n,
\end{align*}
which together with~\eqref{eqn_nec} and~\eqref{eqn_dec} implies that
\begin{align}\label{eqn_er2}
\exists m\neq g(x^n),\;(u_{f(x^n)m}^n(\hat{s}),y^n)\in\mathop{\bigcup}\limits_{s\in\cc{I}(\hat{s})}\cc{T}_{[UY,s]\sigma|\cc{X}|}^n.
\end{align} 

Furthermore, it follows for all $(x^n,y^n)\in\cc{T}_0$ that \mbox{$x^n\in\cc{T}$} and thus for $f(x^n)=i$ and $g(x^n)=j\neq m$, it holds that
\begin{align}\label{eqn_er3}
(u_{i j}^n(\hat{s}),x^n)\in\cc{T}_{[UX,{\hat{s}}]\zeta}^n.
\end{align}
Define the \gls{rv} 
\begin{align*}
\tilde{U}_{\hat{s}}:=\{U_{ij,\hat{s}}^n\}_{(i,j)\in\cc{I}\times\cc{J}}=(U_{11,\hat{s}}^n
,U_{12,\hat{s}}^n,\cdots,U_{|\cc{I}||\cc{J}|,\hat{s}}^n)
\end{align*}
and let $\tilde{u}(\hat{s}):=\{u_{i j}^n(\hat{s})\}_{(i,j)\in\cc{I}\times\cc{J}}$ be an arbitrary realization. 
For all $s\in\cc{I}(\hat{s})$ and $(x^n,y^n)\in\cc{T}_0$, the relations~\eqref{eqn_er2} and~\eqref{eqn_er3} give the following upper-bound for the error probability in~\eqref{eqn_es1}
\begin{align*}
e_s&\big(\tilde{u}(\hat{s})\big):=\sum_{(x^n,y^n)\in\cc{T}_0}\!\!\!\!\rr{P}_{XY,s}^n(x^n,y^n)\nonumber\\
&\times\!\!\!\sum_{\substack{i\in\cc{I},j\in\cc{J}\\m\in\cc{J}-\{j\}}}\!\!\!\!\mathds{1}_{\cc{T}_{[UX,{\hat{s}}]\zeta}^n}(u_{i j}^n(\hat{s}),x^n)\mathds{1}_{\!\!\!\!\mathop{\bigcup}\limits_{s\in\cc{I}(\hat{s})}\!\!\!\!\cc{T}_{[UY,s]\sigma|\cc{X}|}^n}(u_{im}^n(\hat{s}),y^n).
\end{align*}

The upper-bound of $e_s(\tilde{u}(\hat{s}))$ is given by taking its expectation with respect to $\tilde{U}_{\hat{s}}$ as follows
\begin{align}
&\bb{E}_{\tilde{U}_{\hat{s}}}\Big[e_s(\tilde{U}_{\hat{s}})\Big] = \sum_{(x^n,y^n)\in\cc{T}_0}\!\!\!\!\rr{P}_{XY,s}^n(x^n\!\!,y^n)\!\sum_{\substack{i\in\cc{I},j\in\cc{J}\\m\in\cc{J}
-\{j\}}}\!\!\!\bb{E}_{(U_{ij,\hat{s}}^n,U_{\!im,\hat{s}}^n)}\nonumber\\
&\qquad\Big[\mathds{1}_{\cc{T}_{[UX,{\hat{s}}]\zeta}^n}(U_{ij,\hat{s}}^n,x^n)\mathds{1}_{\!\!\mathop{\bigcup}\limits_{s\in\cc{I}(\hat{s})}\!\!\cc{T}_{[UY,s]\sigma|\cc{X}|}^n}(U_{im,\hat{s}}^n,y^n)\Big], \label{eqn_eus}
\end{align}
where the equality follows by \gls{rv}s  $\{U_{ij,\hat{s}}^n\}_{(i,j)\in\cc{I}\times\cc{J}}$ being independent and using the Fubini theorem~\cite[Chapter II, \S 6]{shi_pro}. In the following, an upper-bound for the inner summation in~\eqref{eqn_eus} is derived, which automatically gives the upper-bound for the expectation on the left hand side.

Let $i\in\cc{I},j\in\cc{J},m\in\cc{J}-\{j\}$ and $\tau,\tau'>0$ be given such that $\delta>\tau+\tau'$. It holds for all $(x^n,y^n)\in\cc{T}_0$, all $\zeta,\sigma>0$ sufficiently small, and $n$ sufficiently large that
\begin{align}
&\bb{E}_{(U_{ij,\hat{s}}^n,U_{\!im,\hat{s}}^n)}\Big[\!\mathds{1}_{\cc{T}_{[UX,{\hat{s}}]\zeta}^n}(U_{ij,\hat{s}}^n,x^n)\mathds{1}_{\!\!\!\!\mathop{\bigcup}\limits_{s\in\cc{I}(\hat{s})}\!\!\cc{T}_{[UY,s]\sigma|\cc{X}|}^n}(U_{im,\hat{s}}^n,y^n)\Big]\nonumber\\ 
&=\rr{Pr}\Big(U_{ij,\hat{s}}^n\!\in\!\cc{T}_{[UX,\hat{s}]\zeta}^n(x^n)\Big)\,\rr{Pr}\Big(U_{im,\hat{s}}^n \!\in\!\!\!\!\mathop{\bigcup}\limits_{s\in\cc{I}(\hat{s})}\!\!\!\cc{T}^n_{[UY,s]\sigma|\cc{X}|}(y^n)\Big)\nonumber\\
&\leq |\cc{I}(\hat{s})|\exp\Big[-n\big( \;I(U_{\hat{s}};X_{\hat{s}})\nonumber\\ 
&\qquad\qquad\qquad\qquad+\min\limits_{s\in\cc{I}(\hat{s})}\!\!I(U_{\hat{s}};Y_s) -\tau-\tau' \;\big) \Big],\label{eqn_exp}
\end{align}
where the equality follows again by \gls{rv}s  $\{U_{ij,\hat{s}}^n\}_{(i,j)\in\cc{I}\times\cc{J}}$ being independent. The inequality is a result of Lemma~\ref{lem_typ}.\ref{typ_ixy}. Moreover, the definitions in~\eqref{eqn_n1} and~\eqref{eqn_n2} imply that
\begin{align}
N_{\hat{s},1}&N_{\hat{s},2} (N_{\hat{s},2}-1)\leq N_{\hat{s},1}N_{\hat{s},2}N_{\hat{s},2}\nonumber\\
&=\exp\Big[\;  n\big( I(U_{\hat{s}};X_{\hat{s}}) + \min\limits_{s\in\cc{I}(\hat{s})}\!\!I(U_{\hat{s}};Y_s)-\delta \big) \;\Big].\label{eqn_n12}
\end{align}
By using~\eqref{eqn_exp} and~\eqref{eqn_n12}, it follows that the inner summation in~\eqref{eqn_eus} is upper-bounded by 
$
\big|\cc{I}(\hat{s})\big|\exp\big[-n(\delta-\tau-\tau')\big],
$
which implies that
\begin{equation}\label{eqn_dtt}
\bb{E}_{\tilde{U}_{\hat{s}}}\!\left[e_s(\tilde{U}_{\hat{s}})\right]\leq |\cc{I}(\hat{s})|\exp\big[-n(\delta-\tau-\tau')\big].
\end{equation}

Let $\delta_1$ be sufficiently small such that $0<\delta_1<\delta-\tau-\tau'$. By the Markov inequality, it follows that
\begin{align}
\rr{Pr}\Big( e_s(\tilde{U}_{\hat{s}})\geq\exp(-n\delta_1) \Big)\leq \frac{\bb{E}_{\tilde{U}_{\hat{s}}}\Big[e_s(\tilde{U}_{\hat{s}})\Big]}{\exp(-n\delta_1)}.\label{eqn_mar}
\end{align}

Therefore, \eqref{eqn_dtt} and~\eqref{eqn_mar} imply for all $\zeta$ and $\sigma$ sufficiently small, and all $n$ large enough that
\begin{align}
\rr{Pr}\Big( \!\!\mathop{\bigcap}\limits_{s\in\cc{I}(\hat{s})}\!\!&\Big\{e_s(\tilde{U}_{\hat{s}})\!<\exp(-n\delta_1) \Big\}\Big)\nonumber\\
&\geq 1-|\cc{I}(\hat{s})|^2 \exp(-n(\delta\!-\!\tau\!-\!\tau'\!-\!\delta_1)).\label{eqn_pr1}
\end{align}

Thus, there exists a realization $\tilde{u}(\hat{s})=\{u_{i j}^n(\hat{s})\}_{(i,j)\in\cc{I}\times\cc{J}}$ of the \gls{rv} $\tilde{U}_{\hat{s}}$, where for all $s\in\cc{I}(\hat{s})$ and $(x^n,y^n)\in\cc{T}_0$, the upper-bound of the error probability in~\eqref{eqn_es1} is given by
$\exp(-n\delta_1).$ 
This implies by using~\eqref{eqn_lfn} that the total error probability in~\eqref{eqn_es1} is exponentially small with some $\delta_0>0$.

b) The proof of the second part is very similar to the first part. For a given $(i,j)\!\in\!\cc{I}\times \cc{J}$, the number of chosen random sequences ${v_{pq}^{ij}}^n$ is $N_{\hat{s},3}N_{\hat{s},4}=\exp(n R')$, where $R'$ is the rate of choosing the random sequences. Similar to~\eqref{eqn_rat}, conditions~\eqref{eqn_n3} and~\eqref{eqn_n4} imply that \mbox{$R'=I(V_{\hat{s}};X_{\hat{s}}|U_{\hat{s}})+\delta.$} Furthermore, as a result of Lemma~\ref{lem_typ}.\ref{typ_sub}, from \mbox{$(u_{i j}^n(\hat{s}),x^n)\in\cc{T}^n_{[UX,\hat{s}]\zeta}$} follows \mbox{$\cc{T}^n_{[UVX,\hat{s}]\sigma}(u_{i j}^n(\hat{s}),x^n)\neq\emptyset$}. Therefore, Lemma~\ref{lem_ra2} implies that functions $\varphi$ and $\rho$ as mentioned in the lemma do exist.

According to part a) of this lemma, Bob is able to reconstruct $g(x^n)=j$, by knowing $f(x^n), \hat{s}$ and $y^n$. Therefore, he knows also $u_{i j}^n(\hat{s})$. 
Let $\vartheta\in\bb{R}$ be given such that $\vartheta>\sigma$. Similar to~\eqref{eqn_dec}, the decoder is defined as $\tilde{\rho}(i,j,p,\hat{s},y^n):=q$ if
\begin{align*}
&q\in\cc{Q},\quad {v^{ij}_{pq}}^n(\hat{s}) \!\in\!\!\!\!\mathop{\bigcup}\limits_{s\in\cc{I}(\hat{s})}\!\!\!\!\cc{T}^n_{[UVY,s]\vartheta|\cc{X}|}(u_{i j}^n(\hat{s}),y^n)\quad\text{and }\\
&\forall r\in\cc{Q}, r\neq q\Rightarrow
{v^{ij}_{pr}}^n(\hat{s}) \!\not\in\!\!\!\!\mathop{\bigcup}\limits_{s\in\cc{I}(\hat{s})}\!\!\!\!\cc{T}^n_{[UVY,s]\vartheta|\cc{X}|}(u_{i j}^n(\hat{s}),y^n),
\end{align*}
and otherwise zero.
Define the set $\cc{T}'_0$ as follows. Similar to~\eqref{eqn_lfn}, the probability $\rr{P}_{XY,s}^n({\cc{T}'_0}^{\rr{c}})$ is exponentially small.
\begin{align*}
\cc{T}'_0&:=\Big\lbrace(x^n,y^n)\in\cc{X}^n\times\cc{Y}^n: x^n\in\cc{T}\nonumber\\
&\!\!\!\!\!\!\!\wedge \!(u_{f(x^n)g(x^n)}^n(\hat{s}),{v^{f(x^n)g(x^n)}_{\varphi(x^n)\rho(x^n)}}^n\!(\hat{s}),x^n\!,y^n)
\!\in\!\!\!\!\mathop{\bigcup}\limits_{s\in\cc{I}(\hat{s})}\!\!\!\!\cc{T}^n_{[UVXY,s]\vartheta}\Big\rbrace.
\end{align*}
For a given \mbox{$(i,j)\!\in\!\cc{I}\!\times\!\cc{J}$} define 
$\tilde{V}_{ij,\hat{s}}\!:=\!\{{V^{ij}_{pq,\hat{s}}}^{\!\!\!n}\}_{(p,q)\in\cc{P}\times\cc{Q}}$
and let $\tilde{v}_{ij}(\hat{s}):=\{{v_{pq}^{ij}}^n(\hat{s})\}_{(p,q)\in\cc{P}\times\cc{Q}}$ be an arbitrary realization. 

\noindent Define for all $(x^n,y^n)\in\cc{T}'_0$ with $f(x^n)=i,\, g(x^n)=j$,
\begin{align}
e_s&\big(\tilde{v}_{ij}(\hat{s})\big):=\!\!\!\sum_{\substack{p\in\cc{P},\,q\in\cc{Q}\\r\in\cc{Q}-\{q\}}}\!\mathds{1}_{\!\mathop{\bigcup}\limits_{s\in\cc{I}(\hat{s})}\!\!\!\!\cc{T}_{[UVXY,s]\vartheta}^n}(u_{i j}^n(\hat{s}),{v_{pq}^{ij}}^n(\hat{s}),x^n,y^n)\nonumber\\
&\qquad\qquad\times\mathds{1}_{\!\!\!\!\mathop{\bigcup}\limits_{s\in\cc{I}(\hat{s})}\!\!\!\!\cc{T}_{[UVY,s]\vartheta|\cc{X}|}^n}(u_{i j}^n(\hat{s}),{v_{pr}^{ij}}^n(\hat{s}),y^n).\label{eqn_za2}
\end{align}
Similarly as in part a), it can be shown that the error probability in~\eqref{eqn_es2} is upper-bounded for all $(x^n,y^n)\in\cc{T}'_0$ by
\begin{align}\label{eqn_e55}
\sum_{i\in\cc{I},j\in\cc{J}}\sum_{\substack{(x^n,y^n)\in\cc{T}'_0\\f(x^n)=i,\, g(x^n)=j}}\!\!\!\!\!\!\rr{P}_{XY,s}^n(x^n,y^n)\,e_s\big(\tilde{v}_{ij}(\hat{s})\big).
\end{align}
For $\tau,\tau'>0$ with $\tau+\tau'<\delta$, Lemma~\ref{lem_typ}.\ref{typ_las} implies that 
\begin{align}
&\bb{E}_{({V^{ij}_{pq,\hat{s}}}^{\!\!\!n},{V^{ij}_{\!pr,\hat{s}}}^{\!\!\!n})}\Big[\mathds{1}_{\mathop{\bigcup}
\limits_{s\in\cc{I}(\hat{s})}\cc{T}_{[UVXY,s]\vartheta}^n}(u_{i j}^n(\hat{s}),{V^{ij}_{\!pq,\hat{s}}}^{\!\!\!n},x^n,y^n)\nonumber\\
&\qquad\times\mathds{1}_{\mathop{\bigcup}\limits_{s\in\cc{I}(\hat{s})}\!\!\cc{T}_{[UVY,s]\vartheta|\cc{X}|}^n}(u_{i j}^n(\hat{s}),{V_{pr,\hat{s}}^{ij}}^{\!\!\!n},y^n)\;\big|\; U_{ij,\hat{s}}^n=u_{i j}^n(\hat{s})\Big]\nonumber\\ 
&\leq \big|\cc{I}(\hat{s})\big|^2\cdot \exp\Big[-n\big( \;I(V_{\hat{s}};X_{\hat{s}}|U_{\hat{s}})\nonumber\\ 
&\qquad\qquad\qquad\qquad+\min\limits_{s\in\cc{I}(\hat{s})}\!\!I(V_{\hat{s}};Y_s|U_{\hat{s}}) -\tau-\tau') \;\big) \Big].\label{eqn_za3}
\end{align}
Moreover, by~\eqref{eqn_n3} and~\eqref{eqn_n4}, it follows that
\begin{align}
&N_{\hat{s},3} N_{\hat{s},4} (N_{\hat{s},4}-1)\leq N_{\hat{s},3}N_{\hat{s},4}N_{\hat{s},4}\nonumber\\
&=\exp\Big[\;  n\big( I(V_{\hat{s}};X_{\hat{s}}|U_{\hat{s}}) + \min\limits_{s\in\cc{I}(\hat{s})}\!\!I(V_{\hat{s}};Y_s|U_{\hat{s}})-\delta \big) \;\Big].\label{eqn_za4}
\end{align}
Thus, \eqref{eqn_za3} and~\eqref{eqn_za4} imply by using the definition in~\eqref{eqn_za2} that
\begin{align*}
\bb{E}_{\tilde{V}_{ij,\hat{s}}}\!\Big[e_s(\tilde{V}_{ij,\hat{s}}
)\big| U_{ij,\hat{s}}^n\!=\!u_{i j}^n(\hat{s})\Big]\!\leq\!\big|\cc{I}(\hat{s})\big|^2\!\exp\!\Big[\!-\!n(\delta\!-\!\tau\!-\!\tau')\Big].
\end{align*}
Therefore, it follows for all $0<\delta'_1<\delta-\tau-\tau'$ that
\begin{align*}
\rr{Pr}\Big( \mathop{\bigcap}\limits_{s\in\cc{I}(\hat{s})}&\Big\{e_s(\tilde{V}_{ij,\hat{s}})
<\exp(-n\delta'_1) \Big\}\;\Big|\; U_{ij,\hat{s}}^n\!=\!u_{i j}^n(\hat{s})\Big)\nonumber\\
&\geq 1-\big|\cc{I}(\hat{s})\big|^3\exp\big(-n(\delta-\tau-\tau'-\delta'_1)\big).
\end{align*}
This implies that there exist sequences $\tilde{v}_{ij}(\hat{s})$, for which the error upper-bound in~\eqref{eqn_e55} is exponentially small.
\end{proof}

\begin{lem}\label{lem_skg}
Let Alice's state $\hat{s}\in\hat{\cc{S}}$ be given and $C,D_s,$ and $\hat{S}_{\hat{s}}$ with \mbox{$s\in\cc{I}(\hat{s})$} be \gls{rv}s taking value in $\cc{C},\cc{D},$ and $\hat{\cc{S}}$ respectively. \glspl{rv} $C$ and $D_s$ denote the \gls{cr} and part of Eve's Knowledge respectively. Assume $\alpha\in(0,\frac{1}{6}]$ and $\eta\in(0,\frac{1}{3}]$ with $\alpha\leq\eta$ are given and for all $s\in\cc{I}(\hat{s})$, there exist sets $\cc{B}_s\subset\cc{C}\times\cc{D}$ with 
\begin{align}
\forall (c,d)\in\cc{B}_s,\; \rr{P}_{CD,s|\hat{S}_{\hat{s}}}(c,d|\hat{s})&<(\alpha|\cc{B}_s|)^{-1},\label{eqn_sc1}\\
\rr{P}_{CD,s|\hat{S}_{\hat{s}}}(\cc{B}_s|\hat{s})&\geq 1-(\eta^2-\alpha^2).\label{eqn_sc2}
\end{align}
Furthermore, define the sets $\cc{B}_{s,d}:=\big\{ c\in\cc{C}: (c,d)\in\cc{B}_s \big\}$ and $\cc{D}_s:=\big\{ d\in\cc{D}:\cc{B}_{s,d}\neq\emptyset\big\}$, and assume
\begin{equation}\label{eqn_key}
k\!\in\!\bb{N},\;\, k<\!\alpha^6\!\!\!\!\!\!\!\min_{s\in\cc{I}(\hat{s}),d\in\cc{D}_s}\!\!\!|\cc{B}_{s,d}|,\;\, k<\!\rr{e}^{1/\alpha}(2|\cc{D}|\,|\cc{I}(\hat{s})|)^{-1}.
\end{equation}
Then, there exists a \gls{sk} generator $\kappa:\cc{C}\rightarrow\{1,2,\cdots,k\}$ which maps the \gls{cr} to a \gls{sk} $\kappa(C)$ such that for all $s\!\in\!\cc{I}(\hat{s})$,
\begin{align}\label{eqn_sin}
S\big(\kappa(C)|D_s,\hat{S}_{\hat{s}}=\hat{s}\big)\leq(\alpha+2\eta)\log k + h(\alpha+\eta),
\end{align}
with a probability at least $1-2k\,|\cc{I}(\hat{s})|\,|\cc{D}|\,\rr{e}^{-\frac{\alpha^5 \min |\cc{B}_{s,d}|}{k}}$ where the $\min$ in the exponent is taken over all $s\in\cc{I}(\hat{s})$ and $d\in\cc{D}_s$.
\end{lem}

\begin{proof}[Proof]
Let $s\in\cc{I}(\hat{s})$ be given. Define:
\begin{align}
\lambda&:=\alpha^3\cdot\!\!\!\!\min\limits_{s\in\cc{I}(\hat{s}),d\in\cc{D}_s}|\cc{B}_{s,d}|,\label{eqn_d100}\\
\quad\cc{D}_s'&:=\Big\{ d\in\cc{D}:\rr{P}_{D_s|\hat{S}_{\hat{s}}}(d|\hat{s})\geq\frac{\alpha^2|\cc{B}_{s,d}|}{|\cc{B}_s|} \Big\},\label{eqn_d90}\\\cc{B}_s'&:=\cc{B}_s\cap(\cc{C}\times\cc{D}_s'),\label{eqn_d91}\\
\cc{G}_s&:=\Big\lbrace (c,d)\in\cc{C}\times\cc{D}: \rr{P}_{C|D\hat{S},s}(c|d,\hat{s})\leq\frac{1}{\lambda}  \Big\rbrace,\label{eqn_d95}\\
\cc{G}_{s,d}&:=\Big\{ c\in\cc{C}: (c,d)\in\cc{G}_s \Big\},\label{eqn_d96}\\
\cc{E}_s&:=\Big\{ d\in\cc{D}:\rr{P}_{C|D\hat{S},s}(\cc{G}_{s,d}|d,\hat{s})<1-\eta \Big\}.\label{eqn_d92}
\end{align}

Similarly as in~\cite[Lemma 17.5]{csk_inf}, we show in the first step that the following inequality is true:
\begin{equation}\label{eqn_pde}
\rr{P}_{D_s|\hat{S}_{\hat{s}}}(\cc{E}_s|\hat{s})<\eta.
\end{equation}
This inequality is required later to show that~\eqref{eqn_sin} holds.
For this, let $(c,d)\in\cc{B}_s'$ be given and $s\in\cc{I}(\hat{s})$. It follows that
\begin{align*}
\rr{P}_{\!C|D\hat{S},s}(c|d,\hat{s})&=\frac{\rr{P}_{CD,s|\hat{S}_{\hat{s}}}(c,d|\hat{s})}{\rr{P}_{D_s|\hat{S}_{\hat{s}}}(d|\hat{s})}\leq \frac{(\alpha|\cc{B}_s|)^{-1}}{\alpha^2|\cc{B}_{s,d}||\cc{B}_s|^{-1}}\leq\!\frac{1}{\lambda},
\end{align*}
where the first inequality follows by~\eqref{eqn_sc1},~\eqref{eqn_d90}, and~\eqref{eqn_d91} and the last one by~\eqref{eqn_d100}.
This implies by the definition in~\eqref{eqn_d95} that
\begin{align}\label{eqn_set}
\cc{B}_s'\subset\cc{G}_s\,.
\end{align}
Moreover, by~\eqref{eqn_d90} and~\eqref{eqn_d91} it holds for all $s\in\cc{I}(\hat{s})$ that
\begin{align}\label{eqn_alp}
&\rr{P}_{D_s|\hat{S}_{\hat{s}}}(\cc{D}_s'^{\rr{c}}|\hat{s})=\!\!\sum\limits_{d\in\cc{D}_s'^{\rr{c}}}\!\!\rr{P}_{\!D_s|\hat{S}_{\hat{s}}}\!(d|\hat{s})<\sum_{d\in\cc{D}}\!\frac{\alpha^2|\cc{B}_{s,d}|}{|\cc{B}_s|}\!=\!\alpha^2,\\
&\rr{P}_{CD,s|\hat{S}_{\hat{s}}}(\cc{B}_s\cup(\cc{C}\times\cc{D}_s')|\hat{s})\nonumber\\
& =\rr{P}_{CD,s|\hat{S}_{\hat{s}}}(\cc{B}_s|\hat{s})+\rr{P}_{D_s|\hat{S}_{\hat{s}}}(\cc{D}_s'|\hat{s})-\rr{P}_{CD,s|\hat{S}_{\hat{s}}}(\cc{B}_s'|\hat{s}).\label{eqn_cds}
\end{align}
The relations~\eqref{eqn_alp} and~\eqref{eqn_cds} together with the assumption~\eqref{eqn_sc2} of the lemma imply that
\begin{align}
\rr{P}_{CD,s|\hat{S}_{\hat{s}}}\!(\cc{B}_s'|\hat{s})\!&\geq\rr{P}_{\!CD,s|\hat{S}_{\hat{s}}}\!(\cc{B}_s|\hat{s})-\rr{P}_{\!\!D_s|\hat{S}_{\hat{s}}}\!(\cc{D}_s'^{\rr{c}}|\hat{s})\nonumber\\
&\geq 1-(\eta^2-\alpha^2)-\alpha^2=1-\eta^2.\label{eqn_s11}
\end{align}
By~\eqref{eqn_set} and~\eqref{eqn_s11}, it follows that
$
\rr{P}_{CD,s|\hat{S}_{\hat{s}}}( \cc{G}_s|\;\hat{s} )\geq 1-\eta^2.
$
This implies by the definition in~\eqref{eqn_d96} for all $s\in\cc{I}(\hat{s})$ that
\begin{align*}
1-\eta^2 &\leq \sum_{d\in\cc{D}}\rr{P}_{D_s|\hat{S}_{\hat{s}}}(d|\hat{s})\rr{P}_{C|D\hat{S},s}(\cc{G}_{s,d}|d,\hat{s})\\
&< \sum_{d\in\cc{E}_s}\rr{P}_{D_s|\hat{S}_{\hat{s}}}(d|\hat{s})(1-\eta) +\sum_{d\not\in\cc{E}_s}\rr{P}_{D_s|\hat{S}_{\hat{s}}}(d|\hat{s})\\
&= -\eta\sum_{d\in\cc{E}_s}\rr{P}_{D_s|\hat{S}_{\hat{s}}}(d|\hat{s})+1,
\end{align*}
where the second inequality follows by~\eqref{eqn_d92}.
The desired relation~\eqref{eqn_pde} follows by simplifying this inequality.

In the second step, we show that a \gls{sk} generator $\kappa$, satisfying~\eqref{eqn_sin}, exists.
For this, consider each member of the family of \gls{pd}s $\rr{P}_{C|D\hat{S},s}(.|d,\hat{s})$ with $d\not\in\cc{E}_s$ and $s\in\cc{I}(\hat{s})$. Lemma~\ref{lem_ext} implies for a randomly selected \gls{sk} generator $\kappa$ that
\begin{align}
\rr{Pr}\Big( \big\|\kappa(\rr{P}_{\!C|D\hat{S},s}(.|d,\hat{s}))\!-\!\rr{P}_0\big\|\!>\!2(\alpha+\eta) \Big)\!\!\leq 2k\rr{e}^{-\frac{\lambda\alpha^2}{k}},\label{eqn_s22}
\end{align}
where $\rr{P}_0(i)=1/k$ for all $i=1,2\cdots k$. The universal upper-bound in~\eqref{eqn_s22} was calculated by taking  $\epsilon=2\alpha$ in Lemma~\ref{lem_ext} and by the inequalities $\alpha\leq 1/6$ and $\eta\leq 1/3$ from the assumption of the lemma. Therefore, for the following events 
\begin{align*}
\cc{A}_{s,d}:=\Big\{ \big\|\kappa(\rr{P}_{C|D\hat{S},s}(.|d,\hat{s}))-\rr{P}_0\big\|\leq2(\alpha+\eta) \Big\},
\end{align*}
it follows by~\eqref{eqn_d100} and~\eqref{eqn_s22} that
\begin{align}
\rr{Pr}\Big( \bigcap_{s\in\cc{I}(\hat{s}),d\not\in\cc{E}_s} \!\!\!\cc{A}_{s,d} \Big)&\geq 1-\sum_{s\in\cc{I}(\hat{s}),d\not\in\cc{E}_s}\rr{Pr}\big( \cc{A}_{s,d}^{\rr{c}}\big)\nonumber\\
&\geq 1-2k\,|\cc{D}|\,|\cc{I}(\hat{s})|\rr{e}^{-\frac{\alpha^5\min|\cc{B}_{s,d}|}{k}},\label{eqn_pro}
\end{align}
where the $\min$ in~\eqref{eqn_pro} is taken over all $s\in\cc{I}(\hat{s})$ and $d\in\cc{D}_s$.

This means that a \gls{sk} generator $\kappa$ satisfies the relation $\|\kappa(\rr{P}_{C|D\hat{S},s}(.|d,\hat{s}))-\rr{P}_0\|\leq2(\alpha+\eta)$ simultaneously for all $d\not\in\cc{E}_s$ and $s\in\cc{I}(\hat{s})$ with the probability stated in~\eqref{eqn_pro}. Therefore, it holds by the same probability that
\begin{align*}
S&\big(\kappa(C)|D_s,\hat{S}_{\hat{s}}=\hat{s}\big)\nonumber\\
&=\sum_{d\in\cc{D}}\rr{P}_{\!D_s|\hat{S}_{\hat{s}}}\!(d|\hat{s})\Big[ \!\log k\! - \!H\big(\kappa(C)|D_s=d,\hat{S}_{\hat{s}}=\hat{s}\big) \Big]\\
&\leq \sum_{d\not\in\cc{E}_s}\rr{P}_{D_s|\hat{S}_{\hat{s}}}(d|\hat{s})\big[ (\alpha+\eta)\log k + h(\alpha+\eta) \big]\nonumber\\
&\quad+\sum_{d\in\cc{E}_s}\rr{P}_{D_s|\hat{S}_{\hat{s}}}(d|\hat{s})\log k\leq (\alpha+2\eta)\log k + h(\alpha+\eta),
\end{align*}
which gives the desired relation in~\eqref{eqn_sin}. The equality is the result of Definition~\ref{def_sec} and the first inequality follows from the uniform continuity of entropy~\cite{zha_est},\cite[Problem 3.10]{csk_inf}. The last step is a result of the inequality~\eqref{eqn_pde}.
Moreover, by combining the two assumptions in~\eqref{eqn_key}, it implies that 
\begin{align*}
k\ln\big(\, 2k|\cc{D}|\cdot|\cc{I}(\hat{s})|\,\big)< \alpha^5\cdot\min_{s\in\cc{I}(\hat{s}),d\in\cc{D}_s}|\cc{B}_{s,d}|,
\end{align*}
and consequently, the probability in~\eqref{eqn_pro} is non-zero.
\end{proof}

\subsection{Proof of Main Results}\label{ssc_mai}
In the following, the proofs of Theorems~\ref{thm_uv},~\ref{thm_uci}, and~\ref{thm_inf} are presented. For the proof of the first two theorems, similar techniques which are used for deriving the non-compound \gls{sk} capacity results in~\cite[Theorem 17.21]{csk_inf},\cite{csi_crh}, are used 
and extended to the finite compound setup. For completeness, all proofs are presented in detail. The proof of Theorem~\ref{thm_inf} is based on an 
approximation technique and uses Lemmas~\ref{lem_app} and~\ref{lem_muc}.

\begin{proof}[Proof of Theorem~\ref{thm_uv}]
Let $\delta>0$, and $0<\xi<\zeta<\sigma,$ all in $\bb{R}$ be given. For each $
\hat{s}\in\cc{\hat{S}}=\{\hat{s}_1,\hdots\hat{s}_m\},
$
let $U_{\hat{s}}$ and $V_{\hat{s}}$ satisfy for all $s\in\cc{I}(\hat{s})$ the Markov chains $U_{\hat{s}}-V_{\hat{s}}-X_{\hat{s}}-Y_sZ_s$. 

Consider $N_{\hat{s},1} N_{\hat{s},2}$ sequences $u_{i j}^n(\hat{s})\in\cc{U}^n$ as given in Table~\ref{tab_seu}, which are chosen independently by \gls{pd} $\rr{P}_{U_{\hat{s}}}^n$ with 
\begin{align*}
i\in \cc{I}:=\{1,2,\cdots,N_{\hat{s},1}\},\quad  j\in \cc{J}:=\{1,2,\cdots,N_{\hat{s},2}\},
\end{align*}
and $N_{\hat{s},1}$ and $N_{\hat{s},2}$ satisfying~\eqref{eqn_n1} and~\eqref{eqn_n2} from Lemma~\ref{lem_rac}. 
\begin{table}[!tl]
\centering
\scalebox{1.1}{
\setlength{\tabcolsep}{.4em}
\begin{tabular}{c||c c c c c}
  $\hat{s}_1$ &  $u_{1 1}^n(\hat{s}_1)$ & $\cdots$ & $u_{i j}^n(\hat{s}_1)$ & $\cdots$ & $u_{N_{\hat{s}_1,1}N_{\hat{s}_1,2}}^n(\hat{s}_1)$\\
  \hline
  $\vdots$ & $\vdots$ &  & $\vdots$ &  & $\vdots$\\
  \hline
  $\hat{s}_m$ &  $u_{1 1}^n(\hat{s}_m)$ & $\cdots$ & $u_{i j}^n(\hat{s}_m)$ & $\cdots$ & 
  $u_{N_{\hat{s}_m,1}N_{\hat{s}_m,2}}^n(\hat{s}_m)$\\
\end{tabular}}
\caption{Random sequences for a \gls{dmms} with $|\cc{\hat{S}}|=m$}
\label{tab_seu}
\end{table}
Moreover, for every $u_{i j}^n(\hat{s})$ from Table~\ref{tab_seu}, consider $N_{\hat{s},3} N_{\hat{s},4}$ sequences ${v_{pq}^{ij}}^n(\hat{s})\in\cc{V}^n$ , which are chosen  conditionally independently by \gls{pd} $\rr{P}_{V_{\hat{s}}|U_{\hat{s}}}^n(.|u_{ij}^n(\hat{s}))$ with 
\begin{align*}
p\in \cc{P}:=\{1,2,\cdots,N_{\hat{s},3}\},\quad q\in \cc{Q}:=\{1,2,\cdots,N_{\hat{s},4}\},
\end{align*}
and $N_{\hat{s},3}$ and $N_{\hat{s},4}$ satisfying~\eqref{eqn_n3} and~\eqref{eqn_n4} from Lemma~\ref{lem_rac}. 
Assume that the random sequences $u_{i j}^n(\hat{s})$ in Table~\ref{tab_seu} and their corresponding sequences $\{{v_{pq}^{ij}}^n(\hat{s})\}_{(p,q)\in\cc{P}\times\cc{Q}}$ are known to Alice and Bob. 

To show the achievability of~\eqref{eqn_thm}, the proof is divided into two parts. In part a), the following rate is shown to be achievable:
\begin{align}\label{eqn_ac1}
R_{\rr{sk}}':=\min\limits_{s\in\cc{I}(\hat{s})} I(U_{\hat{s}};Y_s) - \max\limits_{s\in\cc{I}(\hat{s})} I(U_{\hat{s}};Z_s),
\end{align}
when $R_{\rr{sk}}'$ is positive and \gls{rv} $U_{\hat{s}}$ satisfies for all $s\in\cc{I}(\hat{s})$
\begin{align}\label{eqn_scn}
U_{\hat{s}}-X_{\hat{s}}-Y_sZ_s\; \text{ and }\; \max\limits_{s\in\cc{I}(\hat{s})} I(U_{\hat{s}};X_{\hat{s}}|Y_s)<\Gamma.
\end{align}
This gives a special case of~\eqref{eqn_thm} and~\eqref{eqn_cns}. In part b), the achievability of the \gls{sk} rate in~\eqref{eqn_thm} is shown, when it is positive.

\emph{Part a)} Assume $R_{\rr{sk}}'$ from~\eqref{eqn_ac1} is positive i.e.
\begin{align}
\min\limits_{s\in\cc{I}(\hat{s})} I(U_{\hat{s}};Y_s) > \max\limits_{s\in\cc{I}(\hat{s})} I(U_{\hat{s}};Z_s).\label{eqn_paa}
\end{align}

As explained in Section~\ref{sec_cap}, Alice estimates her marginal statistic by hypothesis testing. Assume that $\hat{s}\in\cc{\hat{S}}$ is the index corresponding to the correct decision and all other \mbox{$\tilde{s}\in\cc{\hat{S}}-\{\hat{s}\}$} correspond to a wrong decision. For any observation $X_{\hat{s}}^n$, let the resulting estimated state be denoted by the \gls{rv} $\hat{S}_{\hat{s}}$, taking value in the set $\cc{\hat{S}}$ and having the \gls{pd} $\rr{P}_{\hat{S}_{\hat{s}}}$. It holds by~\cite{hof_tes} and~\cite[Problem 2.13b]{csk_inf} for some $c_0,c_1>0$ that
\begin{align}
\rr{P}_{\hat{S}_{\hat{s}}}(\hat{s})&\geq 1-\exp(-nc_0),\label{eqn_hy1}\\
\forall\tilde{s}\in\hat{\cc{S}}-\{\hat{s}\},\;\rr{P}_{\hat{S}_{\hat{s}}}(\tilde{s})&\leq \exp(-nc_1).\label{eqn_hy2}
\end{align}

Next, Alice sends her estimated marginal source state to Bob over the public noiseless channel. Assume that the hypothesis testing has led to the correct decision $\hat{s}$. Alice and Bob find the corresponding family of sequences $\{u_{i j}^n(\hat{s})\}_{(i,j)\in\cc{I}\times\cc{J}}$ from Table~\ref{tab_seu} by knowing $\hat{s}$. Lemma~\ref{lem_rac}a implies the existence of the encoder functions \mbox{$f:\cc{X}^n\rightarrow\cc{I}\cup\{0\}$} and $g:\cc{X}^n\rightarrow\cc{J}\cup\{0\}$. These encoders give the indices $f(x^n)=i$ and $g(x^n)=j$ of the sequence $u_{i j}^n(\hat{s})$ to be chosen from the family of sequences $\{u_{i j}^n(\hat{s})\}_{(i,j)\in\cc{I}\times\cc{J}}$.
 
As shown in Figure~\ref{fig_cr1}, in addition to the transmitted $\hat{s}$, Alice sends further the index $f(x^n)=i$ to Bob over the public channel. Lemma~\ref{lem_rac}a implies the existence of a decoder
$
\tilde{g}:\cc{I}\times\cc{\hat{S}}\times\cc{Y}^n\rightarrow\cc{J},
$
with which, Bob can reconstruct the index $g(x^n)=j$. This $j$ is the \gls{cr} between Alice and Bob.

In total, for all Alice's estimation results which may lead to a correct or incorrect decision, the error probability upper-bound for all $s\in\cc{I}(\hat{s})$ is given by
\begin{align}
&\rr{Pr}\Big\lbrace g(X_{\hat{s}}^n) \neq \tilde{g}\big(f(X_{\hat{s}}^n),\hat{S}_{\hat{s}},Y_s^n\big)\Big\rbrace \nonumber\\
&=\rr{Pr}\Big\lbrace g(X_{\hat{s}}^n) \neq \tilde{g}\big(f(X_{\hat{s}}^n),\hat{s},Y_s^n\big)\!\wedge\hat{S}_{\hat{s}}\!=\!\hat{s}
\Big\rbrace \nonumber\\&\qquad+\!\!\!\sum_{\tilde{s}\in\hat{\cc{S}}-
\{\hat{s}\}}\!\!\!\!\!\rr{P}_{\hat{S}_{\hat{s}}}(\tilde{s})\rr{Pr}\Big\lbrace g(X_{\hat{s}}^n) \!\neq\! \tilde{g}\big(f(X_{\hat{s}}^n),\hat{S}_{\hat{s}},Y_s^n\big)\,\big|\,\hat{S}_{\hat{s}}\!=\!\tilde{s}\Big\rbrace\nonumber\\
&\leq \exp(-n\delta_0) + \exp(-nc_1)\cdot|\cc{\hat{S}}|,\label{eqn_re1}
\end{align}
where the inequality is a result of~\eqref{eqn_es1} for some $\delta_0>0$ and~\eqref{eqn_hy2}. Thus, condition~\eqref{eqn_ay3} of Definition~\ref{def_ach} is satisfied.

The whole message which is sent over the public channel is represented by the \gls{rv} $f_{\rr{c}}(X_{\hat{s}}^n)=(f(X_{\hat{s}}^n),\hat{S}_{\hat{s}})$ having the range size $\|f\|\cdot|\hat{\cc{S}}|$. As shown in the following, the communication rate satisfies condition~\eqref{eqn_ay1} of Definition~\ref{def_ach}:
\begin{align*}
\frac{1}{n}\log\|f_{\rr{c}}\|&=\frac{1}{n}\log(\|f\|\cdot|\hat{\cc{S}}|)=\frac{1}{n}\log( N_{\hat{s},1} \cdot|\hat{\cc{S}}|)\\&=\max\limits_{s\in\cc{I}(\hat{s})}\! I(U_{\hat{s}};X_{\hat{s}}|Y_s)+3\delta+\frac{1}{n}\log|\hat{\cc{S}}| <\Gamma\!+3\delta,
\end{align*}
where the last equality follows by~\eqref{eqn_n1} and the inequality is a result of~\eqref{eqn_scn} and is valid for all $n$ sufficiently large.

After the index $g(x^n)=j$ is reconstructed by Bob, both Alice and Bob may generate their \gls{sk}, based on this \gls{cr}. Thus, it remains to show that there exists a \gls{sk} generator \mbox{$\kappa:\cc{J}\rightarrow\{1,2\cdots,k\}$}, giving rise to the \gls{rv} $K_{\rr{A}}=\kappa(g(X_{\hat{s}}^n))$, which satisfies condition~\eqref{eqn_ay4} of Definition~\ref{def_ach}.

Again the condition is verified for both estimation results. Assume hypothesis testing has led to the correct decision  and $\hat{s}$ is sent to Bob over the public channel. Define for $s\in\cc{I}(\hat{s})$
\begin{align*}
\cc{T}_s:=\Big\{(x^n,z^n)\in&\,\cc{X}^n\times\cc{Z}^n:\;x^n\in\cc{T}\;\wedge\\ &\;(u_{f(x^n)g(x^n)}^n(\hat{s}),x^n,z^n)\in
\cc{T}^n_{[UXZ,s]\sigma}\Big\},
\end{align*}
where $\cc{T}$ is given in~\eqref{eqn_tfg}.
A similar discussion as for~\eqref{eqn_lfn} from Lemma~\ref{lem_rac}, implies for some $c_2>0$ and $n$ large enough that
\begin{align}\label{eqn_tse}
\rr{P}_{XZ,s}^n(\cc{T}^{\rr{c}}_s)<\exp(-nc_2).
\end{align}

Similarly as in~\cite[Theorem 17.21]{csk_inf} for the non-compound version, define for all $s\in\cc{I}(\hat{s})$, the \glspl{rv} $C$ and $D_s$ and the set $\cc{B}_s$ to be used in Lemma~\ref{lem_skg}, as follows
\begin{align}
C&:=g(X_{\hat{s}}^n),\qquad\qquad D_s:=\big(f(X_{\hat{s}}^n),Z_s^n,\mathds{1}_{\cc{T}_s}(X_{\hat{s}}^n,Z_s^n)\big),\nonumber\\
\cc{B}_s&:=\Big\{ \big(j,(i,z^n,1)\big): (i,j)\in\cc{I}\times\cc{J}, z^n\in\cc{T}_{[Z_s]\xi}^n,\nonumber\\&\qquad\qquad\qquad\qquad\quad\cc{T}_{[UXZ,s]\sigma}^n(u_{i j}^n(\hat{s}),z^n)\neq\emptyset \Big\}.\label{eqn_rvc}
\end{align}

Assume, \glspl{rv} $C$ and $D_s$ take their values in the sets $\cc{C}$ and $\cc{D}$ respectively. Moreover, the sets $\cc{D}_s$ and $\cc{B}_{s,d}$ are defined as in Lemma~\ref{lem_skg}. In the following, it will be shown that all conditions of Lemma~\ref{lem_skg} are satisfied. It holds that
\begin{align}
\rr{P}&_{CD,s|\hat{S}_{\hat{s}}}(\cc{B}_s|\hat{s})=\sum_{(j,(i,z^n,1))\in\cc{B}_s}\rr{P}_{CD,s|\hat{S}_{\hat{s}}}\big(j,(i,z^n,1)|\hat{s}\big)\nonumber\\
&=\sum_{(j,(i,z^n,1))\in\cc{B}_s}\;\,\sum_{\substack{x^n:f(x^n)=i,g(x^n)=j,\nonumber\\ \mathds{1}_{\cc{T}_s}(x^n,z^n)=1}}\rr{P}_{X_{\hat{s}}^n Z_s^n|\hat{S}_{\hat{s}}}(x^n,z^n|\hat{s})\nonumber\\
&=\rr{P}_{X_{\hat{s}}^n Z_s^n|\hat{S}_{\hat{s}}}\!\Big(\cc{T}_s\cap\Big\{ (x^n,z^n)\!\in\!\cc{X}^n\!\times\!\cc{Z}^n: z^n\!\in\!\cc{T}_{[Z_s]\xi}^n\Big\}\,\big|\,\hat{s}\Big)\nonumber\\
&\geq 1-\Big[\rr{P}_{X_{\hat{s}}^n Z_s^n|\hat{S}_{\hat{s}}}(\cc{T}_s^{\rr{c}}|\hat{s})
+\rr{P}_{Z_s^n|\hat{S}_{\hat{s}}}({\cc{T}^n_{[Z_s]\xi}}^{\!\!\!\rr{c}}|\hat{s})\Big]\nonumber\\
&\geq 1-\frac{\rr{P}_{\!X_{\hat{s}}^n Z_s^n}(\cc{T}_s^{\rr{c}}) \!+\!\rr{P}_{Z_s^n}({\cc{T}^n_{[Z_s]\xi}}^{\!\!\!\rr{c}})}{1-\exp(-nc_0)}\geq 1\!-\exp(-nc_3),\label{eqn_de1}
\end{align}
for some $c_3>0$ and $n$ sufficiently large. The last two inequalities follow by~\eqref{eqn_hy1},~\eqref{eqn_tse} and Lemma~\ref{lem_typ}.\ref{typ_ptx}.

\begin{figure}[!t]
\centering
\huge
\scalebox{0.5}{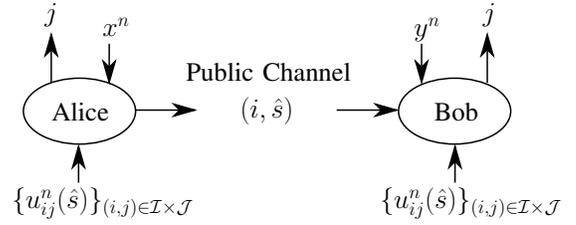}
\caption{Generating the \gls{cr} $j$}
\label{fig_cr1}
\end{figure}

Furthermore, define the parameters $\alpha$ and $\eta$ for some arbitrary $\tau>0$ as follows 
\begin{align}\label{eqn_alf}
\alpha:=\exp(-n(\delta+5\tau)),\quad \eta:=\exp(-n\delta).
\end{align}
For $\delta$ and $\tau$ sufficiently small and $n$ large enough, it holds that $\eta^2-\alpha^2>\exp(-nc_3)$. Therefore, it follows by~\eqref{eqn_de1} that
\begin{align*}
\rr{P}_{CD,s|\hat{S}_{\hat{s}}}(\cc{B}_s|\hat{s})\geq 1-(\eta^2-\alpha^2).
\end{align*}
This guarantees condition~\eqref{eqn_sc2} of Lemma~\ref{lem_skg}. Moreover, the conditions $\alpha\in(0,1/6]$ and \mbox{$\eta\in(0,1/3]$} are also satisfied.

To check condition~\eqref{eqn_sc1} of Lemma~\ref{lem_skg}, we find first an upper-bound of $|\cc{B}_s|$. The non-emptiness constraint $\cc{T}_{[UXZ,s]\sigma}^n(u_{i j}^n(\hat{s}),z^n)\neq\emptyset$ from the definition of the set $\cc{B}_s$ is a sufficient condition for $u_{i j}^n(\hat{s})\in\cc{T}_{[UZ,s]\sigma|\cc{X}|}^n(z^n)$ and thus
\begin{align}
|\cc{B}_s|\leq\!\sum_{z^n\in\cc{T}_{[Z_s]\xi}^n}\!\!\!\Big|\Big\lbrace (i,j)&\in\cc{I}\times\cc{J}:u_{i j}^n(\hat{s})\in\cc{T}_{[UZ,s]\sigma|\cc{X}|}^n(z^n) \Big\rbrace\Big|.\label{eqn_ubd}
\end{align}

Furthermore, for $R$ being the rate of choosing the random sequences $\{u_{i j}^n(\hat{s})\}_{(i,j)\in\cc{I}\times\cc{J}}$, it holds by~\eqref{eqn_n1} and~\eqref{eqn_n2} that      
\begin{align*}
R&=\max_{s\in\cc{I}(\hat{s})}I(U_{\hat{s}};X_{\hat{s}}|Y_s) + \min_{s\in\cc{I}(\hat{s})}I(U_{\hat{s}};Y_s)+\delta\\
&>\min_{s\in\cc{I}(\hat{s})}I(U_{\hat{s}};Y_s)>\max_{s\in\cc{I}(\hat{s})}I(U_{\hat{s}};Z_s),
\end{align*}
where the last inequality is the result of the assumption~\eqref{eqn_paa}. Moreover, for all $(j,(i,z^n,1))\in\cc{B}_s$ it holds that \mbox{$z^n\in\cc{T}_{[Z_s]\xi}^n$} which implies by using Lemma~\ref{lem_typ}.\ref{typ_sub} for $n$ large enough that
$\cc{T}_{[UZ,s]\zeta}^n(z^n)\neq\emptyset.$ 
Therefore, it follows by~\eqref{eqn_ubd} and Lemma~\ref{lem_ran} for $\xi$ and $\zeta$ sufficiently small and $n$ large enough that
\begin{align}
|\cc{B}_s|&\leq \big|\cc{T}_{[Z_s]\xi}^n\big|\exp\Big[ n\big(R-I(U_{\hat{s}};Z_s)+\tau\big) \Big]\nonumber\\
&\leq \exp\Big[ n\big(H(Z_s)+\tau\big) \Big]\nonumber\\ &\qquad\times\exp\Big[ n\big(I(U_{\hat{s}};X_{\hat{s}})+\delta-I(U_{\hat{s}};Z_s)+\tau\big) \Big],\label{eqn_fc1}
\end{align}
where the last inequality follows by~\eqref{eqn_rat} and Lemma~\ref{lem_typ}.\ref{typ_ca1}.

In the second step of verifying condition~\eqref{eqn_sc1}, we find an upper-bound for $\rr{P}_{CD,s|\hat{S}_{\hat{s}}}\big(j,(i,z^n,1)|\hat{s}\big)$. For all $(j,(i,z^n,1))\in\cc{B}_s$, it holds by using~\eqref{eqn_hy1} and~\eqref{eqn_rvc} that
\begin{align}
\big(1-&\exp(-nc_0)\big)\cdot\rr{P}_{CD,s|\hat{S}_{\hat{s}}}\big(j,(i,z^n,1)|\hat{s}\big)\nonumber\\
&\leq \rr{P}_{CD,s}\big(j,(i,z^n,1)\big)\nonumber\\
&\leq\!\!\sum_{x^n\in\cc{T}_{[UXZ,s]\sigma}^n(u_{i j}^n(\hat{s}),z^n)}\!\!\!\!\rr{P}_{XZ,s}^n(x^n,z^n)\nonumber\\
&\leq\big|\cc{T}_{[UXZ,s]\sigma}^n(u_{i j}^n(\hat{s}),z^n)\big|\exp\!\Big[\!-\!n\big(H(X_{\hat{s}}Z_s)\!-\!\tau\big) \Big]\nonumber\\
&\leq\exp\Big[ n\big(H(X_{\hat{s}}|U_{\hat{s}}Z_s) -H(X_{\hat{s}}Z_s)+2\tau\big) \Big],\label{eqn_fc2}
\end{align}
where the third inequality follows by Lemma~\ref{lem_typ}.\ref{typ_px1} and the last one by Lemma~\ref{lem_typ}.\ref{typ_ca2} for $\sigma$ sufficiently small and $n$ large enough. 

\noindent Moreover, by using Markov chains in~\eqref{eqn_scn}, it holds that
\begin{align}
H(X_{\hat{s}}|U_{\hat{s}}Z_s)&-H(X_{\hat{s}}Z_s)\nonumber\\
&=-I(U_{\hat{s}};X_{\hat{s}})-H(Z_s)+I(U_{\hat{s}};Z_s).\label{eqn_fc3}
\end{align}

The relations~\eqref{eqn_fc1},~\eqref{eqn_fc2}, and~\eqref{eqn_fc3} together with the definition of $\alpha$ in~\eqref{eqn_alf} imply for $n$ large enough that
\begin{align*}
|\cc{B}&_s|\,\rr{P}_{CD,s|\hat{S}_{\hat{s}}}\big(j,(i,z^n,1)|\hat{s}\big)\,\alpha\leq\frac{\exp(-n\tau)}{1-\exp(-nc_0)}< 1.
\end{align*}
Thus, condition~\eqref{eqn_sc1} of Lemma~\ref{lem_skg} is also satisfied.

Therefore, Lemma~\ref{lem_skg} implies that there exists a \gls{sk} generator $\kappa:\cc{J}\rightarrow\{1,2,\cdots,k\}$ with $k$ satisfying~\eqref{eqn_key}, such that the relation~\eqref{eqn_sin} holds. By the definitions in~\eqref{eqn_alf}, both $\alpha$ and $\eta$ approach zero exponentially fast. Moreover, by~\eqref{eqn_key} and~\eqref{eqn_fc1}, it follows that $k$ does not increase faster than exponentially. Thus, the inequality~\eqref{eqn_sin} implies for \mbox{$K_{\rr{A}}=\kappa(g(X_{\hat{s}}^n))$} that
\begin{align}
S\big(K_{\rr{A}}\,\big|&\,f(X_{\hat{s}}^n),Z_s^n,\mathds{1}_{\cc{T}_s}(X_{\hat{s}}^n,Z_s^n),\hat{S}_{\hat{s}}=\hat{s}\big)\nonumber\\
&\leq(\alpha+2\eta)\log k + h(\alpha+\eta)\nonumber\\
&\leq n\cdot\exp(-nc_4) +\exp(-nc_4)\leq\exp(-nc_5),\label{eqn_si1}
\end{align}
for some $c_4,c_5>0$. Consequently, the security index is exponentially small for the case of correct estimation decision. 

In total, for all estimation results which may lead to a correct or incorrect decision, the security index is upper-bounded as shown in the following \begin{align*}
S\big(&K_{\rr{A}}\,\big|\,f(X_{\hat{s}}^n),Z_s^n,\mathds{1}_{\cc{T}_s}(X_{\hat{s}}^n,Z_s^n),\hat{S}_{\hat{s}}\big)\nonumber\\
&=\rr{P}_{\hat{S}_{\hat{s}}}(\hat{s})S\big(K_{\rr{A}}|Z_s^n,f(X_{\hat{s}}^n),\mathds{1}_{\cc{T}_s}(X_{\hat{s}}^n,Z_s^n),\hat{S}_{\hat{s}}=\hat{s}\big)\nonumber\\
&+\!\sum_{\tilde{s}\in\hat{\cc{S}}-
\{\hat{s}\}}\!\!\!\!\rr{P}_{\hat{S}_{\hat{s}}}(\tilde{s})S\big(K_{\rr{A}}|Z_s^n,f(X_{\hat{s}}^n),\mathds{1}_{\cc{T}_s}(X_{\hat{s}}^n,Z_s^n),\hat{S}_{\hat{s}}\!=\!\tilde{s}\big)\nonumber\\
&\leq \exp(-nc_5) +n\cdot\exp(-nc_6)\cdot|\cc{\hat{S}}|,
\end{align*}
for some $c_6>0.$ The last inequality is a result of~\eqref{eqn_hy2} and~\eqref{eqn_si1} and the fact that $k$ does not increase faster than exponentially. Therefore, condition~\eqref{eqn_ay4} of Definition~\ref{def_ach} is satisfied.

In the rest of this part, we show that $R_{\rr{sk}}'$ satisfies  condition~\eqref{eqn_ay2} of Definition~\ref{def_ach}.  
By using the definition of $\alpha$ from~\eqref{eqn_alf} and the set $\cc{D}$ being the alphabet of \gls{rv} $D_s$, it follows that the expression $\rr{e}^{1/\alpha}(2|\cc{D}|\,|\cc{I}(\hat{s})|)^{-1}$ from~\eqref{eqn_key} increases doubly exponentially fast.
For any $d:=(i,z^n,1)\in\cc{D}_s$ and $\cc{B}_{s,d}$ from Lemma~\ref{lem_skg}, it follows by~\eqref{eqn_fc1} that $|\cc{B}_{s,d}|$ does not increase faster than exponentially. Therefore, for $n$ large enough, it holds that \begin{align}\label{eqn_kup}
\alpha^6\cdot\min_{s\in\cc{I}(\hat{s}),d\in\cc{D}_s}\big|\cc{B}_{s,d}\big|<\rr{e}^{1/\alpha}(2|\cc{D}|\,|\cc{I}(\hat{s})|)^{-1}.
\end{align} 

Thus, to guarantee condition~\eqref{eqn_key} of Lemma~\ref{lem_skg} it is necessary that $k$ be lower than the left hand side of~\eqref{eqn_kup}. 
For this, a lower-bound for $\min_{s\in\cc{I}(\hat{s}),d\in\cc{D}_s}|\cc{B}_{s,d}|$ should be first determined. 
Lemma~\ref{lem_typ}.\ref{typ_sub} implies that $u_{i j}^n(\hat{s})\in\cc{T}_{[UZ,s]\zeta}^n(z^n)$ is a sufficient condition for $\cc{T}_{[UXZ,s]\sigma}^n(u_{i j}^n(\hat{s}),z^n)\neq\emptyset.$ Therefore for all \mbox{$s\in\cc{I}(\hat{s})$} and $d=(i,z^n,1)\in\cc{D}_s$, it holds that
\begin{equation}\label{eqn_bsd}
\big|\cc{B}_{s,d}\big|\geq\Big|\Big\{  j\in\cc{J}:u_{i j}^n(\hat{s})\in\cc{T}^n_{[UZ,s]\zeta}(z^n) \Big\}\Big|.
\end{equation}

Furthermore by using~\eqref{eqn_n2}, the rate of choosing the random sequences $u_{i j}^n(\hat{s})$ for a fixed index $i$, is given by
\begin{align}\label{eqn_lr1}
R_2=\min_{s\in\cc{I}(\hat{s})}I(U_{\hat{s}};Y_s)-2\delta.
\end{align} 
Thus, for $\delta$ sufficiently small, assumption~\eqref{eqn_paa} gives 
\begin{align}\label{eqn_lr2}
R_2>\max_{s\in\cc{I}(\hat{s})}I(U_{\hat{s}};Z_s). 
\end{align}
Moreover, in deriving the inequality~\eqref{eqn_fc1}, it was shown that for all $(j,d)=\big(j,(i,z^n,1)\big)\in\cc{B}_s$, it holds that $\cc{T}_{[UZ,s]\zeta}^n(z^n)\neq\emptyset$.
Thus, Lemma~\ref{lem_ran} together with~\eqref{eqn_bsd},~\eqref{eqn_lr1}, and~\eqref{eqn_lr2} implies that
\begin{align*}
&\min\limits_{s\in\cc{I}(\hat{s}),d\in\cc{D}_s}\!|\cc{B}_{s,d}| \geq\!\!\!
\min\limits_{s\in\cc{I}(\hat{s}),d\in\cc{D}_s}\Big|\Big\{ j:u_{i j}^n(\hat{s})\in\cc{T}^n_{[UZ,s]\zeta}(z^n) \Big\}\Big|\nonumber\\
&\geq\exp\Big[ n\big( R_2-\max_{s\in\cc{I}(\hat{s})}I(U_{\hat{s}};Z_s)-\tau \big) \Big]\nonumber\\
&\geq\exp\Big[ n\big( \!\min_{s\in\cc{I}(\hat{s})}\!I(U_{\hat{s}};Y_s)\!-\!\!\max_{s\in\cc{I}(\hat{s})}\!I(U_{\hat{s}};Z_s)\!-\!2\delta\!-\!\tau \big) \Big].
\end{align*}
Therefore, by keeping $k$ lower than the left hand side of~\eqref{eqn_kup}, it follows by~\eqref{eqn_alf} and~\eqref{eqn_ac1} that 
$
\frac{1}{n}\log|\cc{K}|> R_{\rr{sk}}'-8\delta.
$
This satisfies condition~\eqref{eqn_ay2} of Definition~\ref{def_ach}.

\begin{figure}[!tl]
\centering
\huge
\scalebox{0.5}{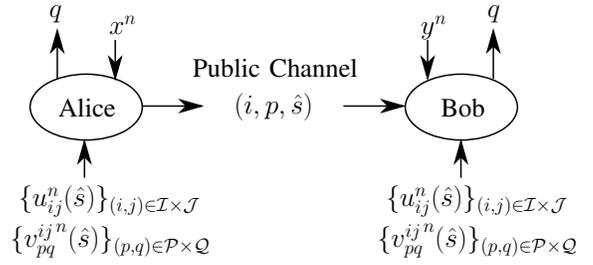}
\caption{Generating the \gls{cr} $q$}
\label{fig_cr2}
\end{figure}

\emph{Part b)} \noindent In this part, it is shown that if
\begin{align*}
R_{\rr{sk}}:=\min_{s\in\cc{I}(\hat{s})}I(V_{\hat{s}};Y_s|U_{\hat{s}})-\max_{s\in\cc{I}(\hat{s})}I(V_{\hat{s}};Z_s|U_{\hat{s}})>0,
\end{align*}
then $R_{\rr{sk}}$ is achievable for all \glspl{rv} $U_{\hat{s}}$ and $V_{\hat{s}}$ satisfying~\eqref{eqn_cns}. 
It may be assumed that
\begin{align}\label{eqn_f99}
\min\limits_{s\in\cc{I}(\hat{s})} I(U_{\hat{s}};Y_s) \leq \max\limits_{s\in\cc{I}(\hat{s})} I(U_{\hat{s}};Z_s).
\end{align}
Because otherwise, it follows by Markov chains in~\eqref{eqn_cns} that 
\begin{align*}
R_{\rr{sk}}
&=\min\limits_{s\in\cc{I}(\hat{s})}\Big[ I(V_{\hat{s}};Y_s)+I(U_{\hat{s}};Y_s|V_{\hat{s}})-I(U_{\hat{s}};Y_s) \Big] \nonumber\\&\quad- \max\limits_{s\in\cc{I}(\hat{s})}\Big[ I(V_{\hat{s}};Z_s)+I(U_{\hat{s}};Z_s|V_{\hat{s}})-I(U_{\hat{s}};Z_s) \Big]\nonumber\\
&\leq  \min\limits_{s\in\cc{I}(\hat{s})}I(V_{\hat{s}};Y_s)    -\max\limits_{s\in\cc{I}(\hat{s})}I(V_{\hat{s}};Z_s)\nonumber\\    &\quad-\Big[ \min\limits_{s\in\cc{I}(\hat{s})}I(U_{\hat{s}};Y_s)    -\max\limits_{s\in\cc{I}(\hat{s})}I(U_{\hat{s}};Z_s) \Big],
\end{align*}
which implies by using part a) that $R_{\rr{sk}}$ is achievable.

Similarly as in part a), Alice sends her estimated marginal source state to Bob over the public noiseless channel.
For the case that hypothesis testing has led to the correct decision $\hat{s}$, both Alice and Bob find the corresponding family of sequences $\{u_{i j}^n(\hat{s})\}_{(i,j)\in\cc{I}\times\cc{J}}$ from Table~\ref{tab_seu} and also the related families of sequences $\{{v_{pq}^{ij}}^n(\hat{s})\}_{(p,q)\in\cc{P}\times\cc{Q}}$ for each member of the chosen row of the table. Lemma~\ref{lem_rac}b implies the existence of the encoder functions \mbox{$\varphi:\cc{X}^n\rightarrow\cc{P}\cup\{0\}$} and \mbox{$\rho:\cc{X}^n\rightarrow\cc{Q}\cup\{0\}$}. These encoders give the indices $\varphi(x^n)=p$ and $\rho(x^n)=q$ of the sequence ${v_{pq}^{ij}}^n(\hat{s})$ to be chosen from the family of sequences $\{{v_{pq}^{ij}}^n(\hat{s})\}_{(p,q)\in\cc{P}\times\cc{Q}}$ for given indices $i$ and $j$. 

As shown in Figure~\ref{fig_cr2}, in addition to the transmitted $\hat{s}$, Alice sends further the indices $f(x^n)=i$ and $\varphi(x^n)=p$ to Bob over the public channel. By part a) of this theorem, $g(x^n)=j$ is known to Bob. Lemma~\ref{lem_rac}b implies the existence of a decoder 
$
\tilde{\rho}:\cc{I}\times\cc{J}\times
\cc{P}\times\cc{\hat{S}}\times\cc{Y}^n\rightarrow\cc{Q}, 
$
with which Bob can reconstruct the index $q$ to be used as the \gls{cr}. The upper-bound of the error probability was given in~\eqref{eqn_es2} from Lemma~\ref{lem_rac}. Due to~\eqref{eqn_f99}, the index $g(x^n)=j$ can not be used as the \gls{cr} any more.

Similarly as in~\eqref{eqn_re1} from part a), for all estimation results which may lead to a correct or incorrect decision, the error probability upper-bound for all $s\in\cc{I}(\hat{s})$ is exponentially small. Therefore, condition~\eqref{eqn_ay3} of Definition~\ref{def_ach} is satisfied.

The public communication function is represented by the \gls{rv} \mbox{$f_{\rr{c}}(X_{\hat{s}}^n)=(f(X_{\hat{s}}^n),\varphi(X_{\hat{s}}^n),\hat{S}_{\hat{s}})$.} As shown in the following, the communication rate satisfies condition~\eqref{eqn_ay1} of Definition~\ref{def_ach}:
\begin{align*}
\frac{1}{n}\log&\big(\|f\|\cdot\|\varphi\|\cdot|\hat{\cc{S}}|\big)=
\max\limits_{s\in\cc{I}(\hat{s})} I(U_{\hat{s}};X_{\hat{s}}|Y_s) +\nonumber\\ &\!\max\limits_{s\in\cc{I}(\hat{s})}I(V_{\hat{s}};X_{\hat{s}}|U_{\hat{s}}Y_s)+6\delta+\frac{1}{n}\log|\hat{\cc{S}}|<\Gamma+6\delta,
\end{align*}
for $n$ sufficiently large. The equality is a result of~\eqref{eqn_n1} and~\eqref{eqn_n3} from Lemma~\ref{lem_rac} and the inequality follows by~\eqref{eqn_cns}.

For showing that conditions~\eqref{eqn_ay2} and~\eqref{eqn_ay4} of Definition~\ref{def_ach} also hold, the \glspl{rv} $C,D_s$, and the set $\cc{B}_s$ can be defined similarly as in part a), but this time for the coding scheme from Lemma~\ref{lem_rac}b). It can be shown that conditions~\eqref{eqn_sc1} and~\eqref{eqn_sc2} of Lemma~\ref{lem_skg} are again satisfied and thus there exists a \gls{sk} generator \mbox{$\kappa:\cc{Q}\rightarrow\{1,2\cdots,k\}$}, giving rise to the \gls{rv} $K=\kappa(\rho(X_{\hat{s}}^n))$, which satisfies conditions~\eqref{eqn_ay2} and~\eqref{eqn_ay4} of Definition~\ref{def_ach}.
The proof follows by a similar discussion as in part a) of this theorem for compound sources. The non-compound version is available in~\cite[Theorem 17.21]{csk_inf}.
\end{proof}

\begin{proof}[Proof of Theorem~\ref{thm_uci}]
For the direct part of the proof, replace $X_{\hat{s}}, Y_s,$ and $Z_s$ in Theorem~\ref{thm_uv} with $X^n_{\hat{s}}, Y^n_s,$ and $Z^n_s$ respectively, for arbitrary $n\in\bb{N}, \hat{s}\in\hat{\cc{S}},$ and $s\in\cc{I}(\hat{s})$. This implies that the \gls{sk} rate
\begin{align*}
\frac{1}{n}\min\limits_{\hat{s}\in\hat{\cc{S}}}
\max_{U_{\hat{s}},V_{\hat{s}}}\Big\{ \min\limits_{s\in\cc{I}(\hat{s})} I(V_{\hat{s}};Y_s^n|U_{\hat{s}}) - \max\limits_{s\in\cc{I}(\hat{s})} I(V_{\hat{s}};Z_s^n|U_{\hat{s}}) \Big\}
\end{align*}
is achievable for any $n\in\bb{N}$, where the outer $\max$ is taken over all \gls{rv}s $U_{\hat{s}}$ and $V_{\hat{s}}$ satisfying~\eqref{eqn_ma2}.

To show the achievability of~\eqref{eqn_tci}, we have to show that the limit in~\eqref{eqn_tci} exists.
Similarly as in~\cite{bje_scw}, the proof follows by using the Fekete's lemma~\cite{fek_sup} which states that if a sequence $a_n$ is superadditive i.e. $a_{n+m}\geq a_n+a_m$, then $\lim\limits_{n\to\infty}a_n/n$ exists.
Define the sequences
\begin{align}
a_n(\hat{s})&:=\max_{U_{\hat{s}},V_{\hat{s}}}\Big\{ \min\limits_{s\in\cc{I}(\hat{s})} I(V_{\hat{s}};Y_s^n|U_{\hat{s}}) - \!\max\limits_{s\in\cc{I}(\hat{s})} I(V_{\hat{s}};Z_s^n|U_{\hat{s}}) \Big\},\nonumber\\
a_n&:=\min_{\hat{s}\in\hat{\cc{S}}}a_n(\hat{s}).\label{eqn_mif}
\end{align}
Take two arbitrary independent Markov chains
\begin{align}
&U_{\hat{s},1}-V_{\hat{s},1}-X_{\hat{s},1}^n-Y_{s,1}^n Z_{s,1}^n\,,\label{eqn_rv1}\\
&U_{\hat{s},2}-V_{\hat{s},2}-X_{\hat{s},2}^m-Y_{s,2}^m Z_{s,2}^m\,,\label{eqn_rv2}
\end{align}
such that $m,n\in\mathbb{N}$ and it holds:
\begin{align*}
&\hat{U}_{\hat{s}}:=\!(U_{\hat{s},1},U_{\hat{s},2}),\; \hat{V}_{\hat{s}}:=\!(V_{\hat{s},1},V_{\hat{s},2}),\; X^{n+m}_{\hat{s}}\!:=\!(X_{\hat{s},1}^n,X_{\hat{s},2}^m),\nonumber\\ 
&Y^{n+m}_s\!:=\!(Y_{s,1}^n,Y_{s,2}^m),\; Z^{n+m}_s\!:=\!(Z_{s,1}^n,Z_{s,2}^m).
\end{align*}

As the two Markov chains are independent, any \gls{rv} from~\eqref{eqn_rv1} is independent of all \glspl{rv} in~\eqref{eqn_rv2} and vice versa. Therefore, the Markov chain $\hat{U}_{\hat{s}}-\hat{V}_{\hat{s}}-X^{n+m}_{\hat{s}}-Y^{n+m}_sZ^{n+m}_s$ holds and
\begin{align*}
&a_{n+m}(\hat{s})\nonumber\\
&=\max_{\hat{U}_{\hat{s}},\hat{V}_{\hat{s}}}\Big\{\!\min\limits_{s\in\cc{I}(\hat{s})} \!I(\hat{V}_{\hat{s}};Y_s^{n+m}|\hat{U}_{\hat{s}}) - \max\limits_{s\in\cc{I}(\hat{s})} \!I(\hat{V}_{\hat{s}};Z_s^{n+m}|\hat{U}_{\hat{s}})\!\Big\}\nonumber\\
&\geq \max_{\hat{U}_{\hat{s}},\hat{V}_{\hat{s}}}\Big\{\!\min\limits_{s\in\cc{I}(\hat{s})} \!I(V_{\hat{s},1};Y_{s,1}^n|U_{\hat{s},1}) - \!\max\limits_{s\in\cc{I}(\hat{s})} \!I(V_{\hat{s},1};Z_{s,1}^n|U_{\hat{s},1})\nonumber\\
 &\quad + \min\limits_{s\in\cc{I}(\hat{s})}\! I(V_{\hat{s},2};Y_{s,2}^m|U_{\hat{s},2}) - \max\limits_{s\in\cc{I}(\hat{s})}\! I(V_{\hat{s},2};Z_{s,2}^m|U_{\hat{s},2})\Big\}\nonumber\\
 &=a_n(\hat{s})+a_m(\hat{s}).
\end{align*}

Thus, by taking the minimum from both sides of this inequality and using the definition in~\eqref{eqn_mif}, it follows that
\begin{align*}
a_{n+m}&\geq \min_{\hat{s}\in\hat{\cc{S}}}\left\{a_n(\hat{s})+a_m(\hat{s})\right\}\nonumber\\
&\geq\min_{\hat{s}\in\hat{\cc{S}}}\{a_n(\hat{s})\}
+\min_{\hat{s}\in\hat{\cc{S}}}\{a_m(\hat{s})\}=a_n+a_m.
\end{align*}
Therefore, the limit in~\eqref{eqn_tci} exists and  is achievable.

For the converse, let $R_{\rr{sk}}>0$ be an achievable \gls{sk} rate and $s\in\cc{S}$ be given. Assume that $\hat{s}\in\hat{\cc{S}}$ is the state of the marginal \gls{rv} $X$ and thus $s\in\cc{I}(\hat{s})$. Alice sends a message $f_{\rr{c}}(X_{\hat{s}}^n)$ to Bob over the public channel and generates a \gls{sk} represented by \gls{rv} $K=K_{\rr{A}}$. It holds by Definition~\ref{def_sec} and condition~\eqref{eqn_ay4} of Definition~\ref{def_ach}, that for all $\delta>0$ and $n$ sufficiently large,
\begin{align}\label{eqn_csi}
\frac{1}{n}\log|\cc{K}|-\frac{1}{n}\min_{s\in\cc{I}(\hat{s})} H\big(K_{\rr{A}}|Z^n_s,f_{\rr{c}}(X^n_{\hat{s}})\big)<\delta.
\end{align}
Moreover, by the Fano's inequality, it holds that
\begin{align}\label{eqn_cfn}
\frac{1}{n}\max_{s\in\cc{I}(\hat{s})}H\big(K_{\rr{A}}|Y^n_s,f_{\rr{c}}(X^n_{\hat{s}})\big)<\frac{1}{n}\delta\log|\cc{K}|+\frac{1}{n}.
\end{align}
By adding~\eqref{eqn_csi} and~\eqref{eqn_cfn}, it follows that
\begin{align*}
\frac{1}{n}\log&|\cc{K}|<\frac{1}{n}\Big[\min_{s\in\cc{I}(\hat{s})} H\big(K_{\rr{A}}|Z^n_s,f_{\rr{c}}(X^n_{\hat{s}})\big)\nonumber\\&-\max_{s\in\cc{I}(\hat{s})}H\big(K_{\rr{A}}|Y^n_s,f_{\rr{c}}(X^n_{\hat{s}})\big)\Big]\!+\!\frac{1}{n}\delta\log|\cc{K}|+\frac{1}{n}+\delta.
\end{align*}

Thus, by using condition~\eqref{eqn_ay2} of Definition~\ref{def_ach}, it follows for $\epsilon=\delta/(1-\delta)+(n(1-\delta))^{-1}+\delta$, and $n$ sufficiently large that
\begin{align*}
&R_{\rr{sk}} < \frac{1}{n}\log|\cc{K}|+\delta\leq\frac{1}{1-\delta}\cdot\frac{1}{n}\Big[\min_{s\in\cc{I}(\hat{s})} I\big(K_{\rr{A}};Y^n_s|f_{\rr{c}}(X^n_{\hat{s}})\big)\nonumber\\
&\qquad\qquad\qquad\qquad\qquad-\max_{s\in\cc{I}(\hat{s})}I\big(K_{\rr{A}};Z^n_s|f_{\rr{c}}(X^n_{\hat{s}})\big)\Big]+\epsilon,
\end{align*}
where the last inequality follows by adding and subtracting the term $H(K_{\rr{A}}|f_{\rr{c}}(X^n_{\hat{s}}))$.

Set \gls{rv}s $U_{\hat{s}}:=f_{\rr{c}}(X^n_{\hat{s}})$ and $V_{\hat{s}}:=(f_{\rr{c}}(X^n_{\hat{s}}),K_{\rr{A}})$.
It holds that $I(X^n_{\hat{s}};U_{\hat{s}}|V_{\hat{s}})=0$. Furthermore, as $f_{\rr{c}}(X^n_{\hat{s}})$ and $K_{\rr{A}}$ are both functions of $X^n_{\hat{s}}$, it implies that $I(Y_s^n,Z_s^n\,;U_{\hat{s}}V_{\hat{s}}|X^n_{\hat{s}})=0$. This proves that the Markov chains in~\eqref{eqn_ma2} are valid.

Finally, by taking the maximum with respect to $U_{\hat{s}}$ and $V_{\hat{s}}$ and the minimum with respect to $\hat{s}\in\cc{\hat{S}}$, it follows for $\delta>0$ sufficiently small and $n$ large enough  that
\begin{align*}
R_{\rr{sk}}\!\leq\frac{1}{n}&\min_{\hat{s}\in\hat{S}}\max_{U_{\hat{s}},V_{\hat{s}}}\!\Big[\min_{s\in\cc{I}(\hat{s})} \!\!I(V_{\hat{s}};Y^n_s|U_{\hat{s}})\!-\!\!\max_{s\in\cc{I}(\hat{s})}\!\!I(V_{\hat{s}};Z^n_s|U_{\hat{s}})\Big],
\end{align*}
which completes the proof.
\end{proof}

\begin{proof}[Proof of Theorem~\ref{thm_inf}]
The achievability of the \gls{sk} rate in~\eqref{eqn_if2} for a finite source follows directly from the special case of Theorem~\ref{thm_uv}. 
By taking \gls{rv} $U_{\hat{s}}$ to be $X_{\hat{s}}$ in part a) of the proof of Theorem~\ref{thm_uv} in~\eqref{eqn_ac1}, it follows that the given rate in~\eqref{eqn_if2} for a finite $\cc{S}$ is achievable.
 
To show the achievability of the \gls{sk} rate in~\eqref{eqn_if2} for an infinite set of source states $\cc{S}$ and a finite set of marginal states $\cc{\hat{S}}$, let $\hat{s}\in\cc{\hat{S}}$ be given and consider the infinite class of stochastic matrices as follows:
\begin{align}
\{\rr{P}_{YZ,s|X_{\hat{s}}}:\cc{X}\to\cc{P(Y}\times \cc{Z})\}_{s\in\cc{I}(\hat{s})}.\label{eqn_sti}
\end{align}

Similarly as in~\cite{sch_gau}, Lemma~\ref{lem_app} implies that for any \mbox{$l>2|\cc{Y\times Z}|^2,$} there exists a finite set of stochastic matrices
\begin{align*}
\{\rr{W}_{s'}\!:\!\cc{X}\!\!\to\!\cc{P(Y}\!\times \!\cc{Z})\}_{s'\in\cc{I}'(\hat{s})},\text{ with } |\cc{I}'(\hat{s})|\!\leq\!(l\!+\!1)^{|\cc{X\times Y\times Z}|},
\end{align*}
which approximates the one in~\eqref{eqn_sti} such that
\begin{align}
&\forall s\in\cc{I}(\hat{s}),\;\; \exists s'\in\cc{I}'(\hat{s}),\;\; \forall (x,y,z)\in\cc{X \times Y\times Z},\nonumber\\
&\big|\rr{P}_{YZ,s|X_{\hat{s}}}(y,z|x)-\rr{W}_{s'}(y,z|x)\big|\leq\frac{1}{l}|\cc{Y\!\times\! Z}|,\label{eqn_bl1}\\
&\rr{P}_{YZ,s|X_{\hat{s}}}(y,z|x)\leq \rr{e}^{\frac{2|\cc{Y\times Z}|^2}{l}}\rr{W}_{s'}(y,z|x).\label{eqn_bl2}
\end{align}

As $\hat{s}\in\cc{\hat{S}}$ was chosen arbitrarily, we may repeat this procedure for all $\hat{s}\in\cc{\hat{S}}$ and define the following finite source
\begin{align*}
\mathfrak{S}':=\{XYZ,s'\}_{s'\in\cc{I}'(\hat{s}),\hat{s}\in\hat{\cc{S}}}\,,
\end{align*}
where for all $\hat{s}\in\cc{\hat{S}},s'\in\cc{I}'(\hat{s}),$ and $(x,y,z)\in\cc{X \times Y\times Z}$, the \gls{pd} of $\mathfrak{S}'$ is given by
\begin{align*}
\rr{P}_{XYZ,s'}(x,y,z):=\rr{P}_{X_{\hat{s}}}(x)\rr{W}_{s'}(y,z|x).
\end{align*}
Define for a marginal index $\hat{s}$ and indices $r,t\in\cc{I}(\hat{s})$,
\begin{align}
f(\hat{s},r,t)&:=I(X_{\hat{s}};Y_r)-I(X_{\hat{s}};Z_t).\label{eqn_fde}
\end{align}
Since $\mathfrak{S}'$ is a finite source, the following \gls{sk} rate is again by using Theorem~\ref{thm_uv} achievable:
\begin{align*}
\min_{\hat{s}\in\cc{\hat{S}}}
\min_{r',t'\in\cc{I}'(\hat{s})}\!f(\hat{s},r',t').
\end{align*}

In the following, it is shown that the \gls{sk} generation protocol which guarantees the achievability of this \gls{sk} rate
for the finite source $\mathfrak{S}'$, also guarantees the achievability of the \gls{sk} rate given in~\eqref{eqn_if2}
for the infinite source $\mathfrak{S}$ when $l=n^3$.

By using~\eqref{eqn_bl1} and taking $l=n^3$, it follows for all $r,t\in\cc{I}(\hat{s})$ and their corresponding indices $r',t'\in\cc{I}'(\hat{s})$ that
\begin{align*}
&\gamma_1:=\frac{1}{2}\Vert \rr{P}_{XY,r}-\rr{P}_{XY,r'}\Vert\leq\frac{1}{2n^3}|\cc{Y}\times \cc{Z}|^2,\nonumber\\
&\gamma_2:=\frac{1}{2}\Vert \rr{P}_{XZ,t}-\rr{P}_{XZ,t'}\Vert\leq\frac{1}{2n^3}|\cc{Y}\times \cc{Z}|^2.
\end{align*}

Since $\gamma_1\leq 1-(|\cc{X}||\cc{Y}|)^{-1}$ and $\gamma_2\leq 1-(|\cc{X}||\cc{Z}|)^{-1}$ hold for $n$ large enough, 
it follows by Lemma~\ref{lem_muc} that
\begin{align}
|I(X_{\hat{s}};Y_r)-&I(X_{\hat{s}};Y_{r'})|
\leq 3\gamma_1\log(|\cc{X\times Y}|)+3h(\gamma_1),\label{eqn_mu1}\\
|I(X_{\hat{s}};Z_t)-&I(X_{\hat{s}};Z_{t'})|\leq 3\gamma_2\log(|\cc{X\times Z}|)+3h(\gamma_2).\label{eqn_mu2}
\end{align}
By using~\eqref{eqn_fde},~\eqref{eqn_mu1}, and~\eqref{eqn_mu2}, it holds that
\begin{align}
&|f(\hat{s},r,t)-f(\hat{s},r',t')|\nonumber\\
&\leq |I(X_{\hat{s}};Y_r)-I(X_{\hat{s}};Y_{r'})|+|I(X_{\hat{s}};Z_t)-I(X_{\hat{s}};Z_{t'})|\nonumber\\
&\leq  \frac{3|\cc{Y}\times \cc{Z}|^2}{2n^3}\log(|\cc{X}^2\times \cc{Y}\times \cc{Z}|)+6h\Big(\frac{|\cc{Y}\times \cc{Z}|^2}{2n^3}\Big).\label{eqn_ep1}
\end{align}

As stated before, the protocol related to the finite source $\mathfrak{S}'$, is used for the infinite source  $\mathfrak{S}$. To show that the \gls{sk} rate of this protocol is close to the one which is given in~\eqref{eqn_if2}, it is sufficient to show that for a given $\epsilon>0$ it holds:
\begin{align}
\big|\min_{\hat{s}\in\hat{\cc{S}}}\inf_{r,t\in\cc{I}(\hat{s})}f(\hat{s},r,t)-\min_{\hat{s}\in\cc{\hat{S}}}
\min_{r',t'\in\cc{I}'(\hat{s})}\!f(\hat{s},r',t')\big|<\epsilon.\label{eqn_ra1}
\end{align}
Let $\tau>0$ be given. There exist then $r_0,t_0\in\cc{I}(\hat{s})$ such that
\begin{align*}
\inf_{r,t\in\cc{I}(\hat{s})}f(\hat{s},r,t)&>f(\hat{s},r_0,t_0)-\tau\nonumber\\
&\geq\min_{r',t'\in\cc{I}'(\hat{s})}\!f(\hat{s},r',t')-\epsilon,
\end{align*}
where the second inequality holds for $\tau$ small enough and is a result of~\eqref{eqn_ep1} and the concavity of entropy. 
Similarly, there exist \mbox{$r_0',t_0'\in\cc{I}'(\hat{s})$} such that
\begin{align*}
\min_{r',t'\in\cc{I}'(\hat{s})}f(\hat{s},r',t')&\geq f(\hat{s},r'_0,t'_0)-\tau\nonumber\\&>\inf_{r,t\in\cc{I}(\hat{s})}\!f\big(\hat{s},r,t\big)-\epsilon.
\end{align*}
Therefore, as the index $\hat{s}\in\hat{\cc{S}}$ was taken arbitrarily, the relation~\eqref{eqn_ra1} follows directly.
This implies that condition~\eqref{eqn_ay2} of Definition~\ref{def_ach} is satisfied for the infinite source $\mathfrak{S}$.

For verifying condition~\eqref{eqn_ay3} of Definition~\ref{def_ach}, the following inequality is required which follows by using~\eqref{eqn_bl2}:
\begin{align}
\rr{P}_{Y_s|X_{\hat{s}}}^n(y^n|x^n)&=\prod_{i=1}^n \rr{P}_{Y_s|X_{\hat{s}}}(y_i|x_i)\nonumber\\&\leq\rr{e}^{\frac{2|\cc{Y\times Z}|^2n}{l}}\rr{V}_{s'}^n(y^n|x^n),\label{eqn_prb}
\end{align} 
where $\rr{V}_{s'}(y|x):=\sum_{z\in\cc{Z}}\rr{W}_{s'}(y,z|x)$ for any $x\in\cc{X},y\in\cc{Y}$.

In order to use the finite protocol for the infinite source $\mathfrak{S}$, the \gls{pd} of the marginal \gls{rv} $X$ is estimated by Alice. Let the correct estimation decision be given by $\hat{s}$. The result of estimation for the source $\mathfrak{S}$ is identical with the one from the source $\mathfrak{S}'$. Moreover, assume that the encoder $g$ and decoder $\tilde{g}$ which are given by Lemma~\ref{lem_rac}a), are used to generate the \gls{cr}. The error probability upper-bound for all $s\in\cc{I}(\hat{s})$ is given by
\begin{align*}
&\rr{Pr}\Big\lbrace g(X_{\hat{s}}^n) \neq \tilde{g}\big(f(X_{\hat{s}}^n),\hat{S}_{\hat{s}},Y_s^n\big)\Big\rbrace \nonumber\\
&\leq\rr{P}_{XY,s}^n\big(\big\{(x^n,y^n)\in\cc{X}^n\!\times\!\cc{Y}^n:g(x^n)\neq \tilde{g}(f(x^n),\hat{s},y^n) \big\}\big)\\&\quad\;\;+\!\!\!\sum_{\tilde{s}\in\hat{\cc{S}}-
\{\hat{s}\}}\!\!\!\!\rr{P}_{\!\hat{S}_{\hat{s}}}(\tilde{s})\rr{Pr}
\Big\lbrace g(X_{\hat{s}}^n)\neq \tilde{g}\big(f(X_{\hat{s}}^n),\hat{S}_{\hat{s}},Y_s^n\big)\big|\hat{S}_{\hat{s}}=\tilde{s}\Big\rbrace\nonumber\\
&\leq \rr{e}^{\frac{2|\cc{Y\times Z}|^2}{l}n}\times\rr{P}_{X_{\hat{s}}}^n(x^n)\times\rr{V}_{s'}^n
\big(\big\{(y^n\in\cc{Y}^n\!:\nonumber\\&\qquad\;g(x^n)\neq \tilde{g}(f(x^n),\hat{s},y^n) \big\}\big|\,x^n\big)+ \exp(-nc_1)\cdot|\cc{\hat{S}}|\nonumber\\
&\leq \rr{e}^{\frac{2|\cc{Y\times Z}|^2}{l}n}\times\rr{P}_{XY,s'}^n
\big(\big\{(x^n,y^n)\in\cc{X}^n\!\times\!\cc{Y}^n:\nonumber\\
&\qquad g(x^n)\neq \tilde{g}(f(x^n),\hat{s},y^n) \big\}\big)+ \exp(-nc_1)\cdot|\cc{\hat{S}}|\nonumber\\
&\leq \rr{e}^{\frac{2|\cc{Y\times Z}|^2}{n^2}}\exp(-n\delta_0) + \exp(-nc_1)\cdot|\cc{\hat{S}}|,\nonumber
\end{align*}
for some $c_1,\delta_0>0$. The second inequality follows by~\eqref{eqn_prb} and~\eqref{eqn_hy2}. The last inequality is a result of~\eqref{eqn_es1} from Lemma~\ref{lem_rac}a) and that $l=n^3$. Therefore, the error probability is exponentially small.
Furthermore, as $|\cc{I}'(\hat{s})|<(l+1)^{|\cc{X\times Y\times Z}|}$ and $l=n^3$, a universal $\delta_1$ exists for which the probability in~\eqref{eqn_pr1} is non-zero and thus such a coding scheme for the finite source $\mathfrak{S}'$ exists.

To show that condition~\eqref{eqn_ay4} of Definition~\ref{def_ach} is also guaranteed, assume again that the encoder $g$ and decoder $\tilde{g}$ from Lemma~\ref{lem_rac}a) are used and the correct estimation decision is $\hat{s}$. As the protocol guarantees condition~\eqref{eqn_ay4} of Definition~\ref{def_ach} for the finite source $\mathfrak{S}'$, it holds by using Definitions~\ref{def_sec} for all $\epsilon_0>0$, $n$ sufficiently large, and $s'\in\cc{I}'(\hat{s})$ that
\begin{align}
S(K_{\rr{A}}|Z_{s'}^n,f&(X_{\hat{s}}^n),\hat{S}_{\hat{s}}\!=\!\hat{s})=\log|\cc{K}|-H(\kappa(g(X_{\hat{s}}^n))|\hat{S}_{\hat{s}}\!=\!\hat{s})\nonumber\\
&+\!I(\kappa(g(X_{\hat{s}}^n));Z_{s'}^n,f(X_{\hat{s}}^n)|\hat{S}_{\hat{s}}\!=\!\hat{s})< \epsilon_0,\label{eqn_se1}
\end{align}
where $K_{\rr{A}},\cc{K},$ and $\kappa$ are given by Definition~\ref{def_pro} and Lemma~\ref{lem_skg}. 
Furthermore, it holds that
\begin{align}
&\gamma_3:=\!\frac{1}{2}\big\Vert\rr{P}_{\!\!\kappa(g(X_{\hat{s}}^n))Z_s^n f(X_{\hat{s}}^n)|\hat{S}_{\hat{s}}}\!(\cdot|\hat{s})\!-\!\rr{P}_
{\!\!\kappa(g(X_{\hat{s}}^n))Z_{s'}^n f(X_{\hat{s}}^n)|\hat{S}_{\hat{s}}}\!(\cdot|\hat{s})\big \Vert\nonumber\\
&\leq \frac{1}{2\!-\!2\exp(-nc_0)}\big\Vert\rr{P}_{\!\!\kappa(g(X_{\hat{s}}^n))Z_s^n f_{\rr{c}}(X_{\hat{s}}^n)}\!-\!\rr{P}_
{\!\!\kappa(g(X_{\hat{s}}^n))Z_{s'}^n f_{\rr{c}}(X_{\hat{s}}^n)}\big \Vert\nonumber\\
&\leq \frac{n\big\Vert \rr{P}_{XZ,s} - \rr{P}_{XZ,s'} \big\Vert}{2\!-\!2\exp(-nc_0)} \leq\frac{|\cc{Y}\times \cc{Z}|^2}{2n^2(1\!-\!\exp(-nc_0))}\,,\label{eqn_gm3} 
\end{align}
where \glspl{rv} $f(X_{\hat{s}})$ with alphabet $\cc{I}$ and $f_{\rr{c}}(X_{\hat{s}}^n)=(f(X_{\hat{s}}^n),\hat{S}_{\hat{s}})$ are given by Lemma~\ref{lem_rac}a) and proof of Theorem~\ref{thm_uv}.
The first inequality follows by~\eqref{eqn_hy1} and the second one by the fact that no mapping of the \glspl{pd} increases their 1-norm distance. The last inequality is a result of~\eqref{eqn_bl1} and that $l=n^3$.
Since $\gamma_3\leq 1- \frac{1}{|\cc{K}\times\cc{Z}^n\times\cc{I}|}$, then for all $\epsilon_1>0$, Lemma~\ref{lem_muc} implies by using~\eqref{eqn_gm3} that
\begin{align}
&\big|I(\kappa(g(X_{\hat{s}}^n));Z_{s}^n,f(X_{\hat{s}}^n)|\hat{S}_{\hat{s}}\!=\!\hat{s})\nonumber\\&\qquad\qquad\qquad\qquad\qquad-I(\kappa(g(X_{\hat{s}}^n));Z_{s'}^n,f(X_{\hat{s}}^n)|\hat{S}_{\hat{s}}\!=\!\hat{s})\big|\nonumber\\
&\leq 3\gamma_3\log(|\cc{K}\times\cc{Z}^n\times\cc{I}|-1)+
3h(\gamma_3)\nonumber\\
&\leq \frac{3|\cc{Y}\times \cc{Z}|^2 c_2}{2n(1\!\!-\!\exp(\!-nc_0))}\!+\!3h\Big(\frac{|\cc{Y}\times \cc{Z}|^2}{2n^2(1\!\!-\!\exp(\!-nc_0))}\Big)\!\!<\!\epsilon_1,\label{eqn_ep5}
\end{align}
for some $c_2>0$ and $n$ sufficiently large. The second inequality follows by the fact that the argument of the log function does not increase faster than exponentially. 
Therefore, it follows by using~\eqref{eqn_se1} and~\eqref{eqn_ep5} that for all $s\in\cc{I}(\hat{s})$,
\begin{align*}
S(K_{\rr{A}}|Z_s^n,f(X_{\hat{s}}^n),\hat{S}_{\hat{s}}\!=\!\hat{s})<\epsilon_0+\epsilon_1.
\end{align*}

This inequality together with~\eqref{eqn_hy2} implies for all estimation results which may lead to a correct or incorrect decision that
\begin{align*}
S(K_{\rr{A}}|Z_s^n,f(X_{\hat{s}}^n),\hat{S}_{\hat{s}})&= \rr{P}_{\hat{S}_{\hat{s}}}(\hat{s})S(K_{\rr{A}}|Z_s^n,f(X_{\hat{s}}^n),\hat{S}_{\hat{s}}\!=\!\hat{s})\nonumber\\
+\sum_{\tilde{s}\in\hat{\cc{S}}-
\{\hat{s}\}}\!\!\!\!&\rr{P}_{\hat{S}_{\hat{s}}}(\tilde{s})S\big(K_{\rr{A}}|Z_s^n,f(X_{\hat{s}}^n),\hat{S}_{\hat{s}}\!=\!\tilde{s}\big)\nonumber\\&\leq \epsilon_0+\epsilon_1+|\cc{\hat{S}}|\exp(-nc_1).
\end{align*}

The existence of the \gls{sk} generator $\kappa$ for the finite source $\mathfrak{S}'$
follows by applying the relation $|\cc{I}'(\hat{s})|<(n^3+1)^{|\cc{X\times Y\times Z}|}$ to~\eqref{eqn_key} and~\eqref{eqn_kup}, 
so that the probability in~\eqref{eqn_pro} is non-zero.

For the converse proof, assume that $R_{\rr{sk}}$ is achievable. The result from~\cite[Theorem 1]{ahl_crm} implies that
\begin{align}
R_{\rr{sk}}\leq\min_{\hat{s}\in\hat{S}}\inf_{\,r,t\in\cc{I}(\hat{s})}I(X_{\hat{s}};Y_r|Z_t).\label{eqn_in2}
\end{align}
Furthermore, for any given $\hat{s}\in\cc{\hat{S}}$ and all $r,t\in\cc{I}(\hat{s})$, it holds by~\eqref{eqn_if1} that
\begin{align*}
I(X_{\hat{s}};Y_r)-I(X_{\hat{s}};Z_t)=I(X_{\hat{s}};Y_r|Z_t).
\end{align*}
This identity together with~\eqref{eqn_in2} completes the proof.
\end{proof}

\section{Conclusion}\label{sec_con}

The \gls{sk} generation protocol which was introduced in this work used a two phase approach to achieve the given \gls{sk} rate. In the first step, Alice estimated her state and sent this along with other information which was obtained from her observation to Bob. Although, this information is also received by Eve, it was shown that the strong secrecy and uniformity of the generated \gls{sk} is still guaranteed. In the second step, Bob used this information including the estimated state of Alice to generate the \gls{sk}. A single-letter lower-bound for the \gls{sk} capacity of a finite compound source was derived as a function of the communication rate constraint between Alice and Bob.
This result was further extended to a multi-letter \gls{sk} capacity formula by discarding the public communication rate constraint. As the final result, a single-letter \gls{sk} capacity formula was derived for degraded compound sources with no communication constraint and an arbitrary set of source states. It was shown that for any infinite compound source, with finite marginal set of states, there exists an approximating finite source whose \gls{sk} generation protocol also guarantees the achievability of the given rates for the infinite source.

\bibliographystyle{IEEEtran} 
\bibliography{IEEEabrv,refs.bib}

\end{document}